\documentclass[11pt]{article}
\usepackage{amssymb}

\usepackage{graphicx}
\usepackage{amsmath}
\usepackage{makeidx}
\usepackage{indentfirst}

\newcounter{resultnum}[section]

\newcounter{conclusionnum}[section]

\newcounter{conditionnum}[section]

\newcounter{conjecturenum}[section]

\newcounter{examplenum}[section]

\newcounter{exercisenum}[section]

\newcounter{lemmanum}[section]

\newcounter{notationnum}[section]
\newtheorem{theorem}{Theorem}[section]

\newcounter{theoremnum}[section]

\newcounter{definitionnum}[section]

\newcounter{corollarynum}[section]

\newcounter{remarknum}[section]

\newcounter{propositionnum}[section]

\newcounter{acknowledgementnum}[section]

\newcounter{algorithmnum}[section]

\newcounter{axiomnum}[section]

\newcounter{casenum}[section]

\newcounter{claimnum}[section]

\newcounter{summarynum}[section]

\newcounter{problemnum}[section]
\newenvironment{proof}[1][]{\textbf{Proof.} }{}

\begin{document}

\title{Finsler Branes and Quantum Gravity Phenomenology with Lorentz
Symmetry Violations}
\date{July 5, 2011}
\author{\textbf{Sergiu I. Vacaru} \thanks{%
sergiu.vacaru@uaic.ro, Sergiu.Vacaru@gmail.com} \\
\textsl{\small University "Al. I. Cuza" Ia\c si, Science Department,} \\
\textsl{\small 54 Lascar Catargi street, Ia\c si, Romania, 700107 } }
\maketitle

\begin{abstract}
A consistent theory of quantum gravity (QG) at Planck scale almost sure
contains manifestations of Lorentz local symmetry violations (LV) which may be detected at observable scales. This can be effectively described and classified by models with nonlinear dispersions and related Finsler metrics
and fundamental geometric objects (nonlinear and linear connections)
depending on velocity/ momentum variables. We prove that the trapping brane mechanism provides an accurate description of gravitational and matter field phenomena with LV over a wide range of distance scales and recovering in a systematic way the general relativity (GR) and local Lorentz symmetries. In contrast to the models with extra spacetime dimensions, the Einstein--Finsler type gravity theories are positively with nontrivial
nonlinear connection structure, nonholonomic constraints and torsion induced by generic off--diagonal coefficients of metrics, and determined by
fundamental QG and/or LV effects.

\vskip0.1cm

\textbf{Keywords:}\ quantum gravity, Lorentz violation, nonlinear
dispersion, Finsler geometry, brane physics.

\vskip3pt

PACS:\ 02.40.-k, 04.50.Kd, 04.60.Bc, 04.70.-s, 04.90.+e, 11.25.-w, 11.30.Cp,
98.80.Cq
\end{abstract}


\section{Introduction and Preliminaries}

There are several reasons to study generalizations of the Einstein gravity
theory to models with local anisotropy, extra dimensions, analogous
gravitational interactions and Finsler geometries. The first one goes in
relation to the so--called quantum gravity (QG) phenomenology being
supported by a number of ideas and research on possible observable QG
effects and induced violations of Lorentz invariance (LV), see recent
reviews \cite{kost4,xiao,liberati,mavromatos1}.

There were also analyzed possible production QG scenarios of mini--black
holes in TeV--scale at colliders \cite{dimopoulos}, or in cosmic rays \cite%
{anchor}, and Planck--scale fuziness of spacetime \cite{amelino}. For some
special situations, the QG effects/objects may manifest as a
non--commutative geometry \cite{carroll} or on some brane--world backgrounds %
\cite{burgess}. A series of tentative results in various approaches to QG
and string theory, and computations of local quantum field theory, suggest
the proposal that the Lorenz invariance may be only a low energy symmetry.

Let us consider two important arguments about nonlinear dispersion relations
in QG and possible related modifications of the fundamental concepts on
spacetime geometry:

\begin{enumerate}
\item \textbf{Nonlinear dispersions and LV in QG. \ }Generically, we can
write for a particle of mass $m_{0}$ propagating in a ''slightly deformed''
four dimensional (4--d) Minkowski spacetime
\begin{equation}
E^{2}=p^{2}c^{2}+m_{0}^{2}c^{4}+\varphi (E,p;\mu ;M_{P}),  \label{defep}
\end{equation}%
where $c$ is the light velocity, $E$ and $p$ are respectively the energy and
momentum of the particle; $\mu $ is some particle physics mass scale and
(normally) assumes that the Planck mass $M_{P}\approx 1.22\times 1019GeV$
denotes the mass scale at which the QG corrections become appreciable. The
nonlinear term $\varphi (...)$ encodes possible quantum matter and gravity
effects and LV terms\footnote{%
in explicit form, we have to consider additional dependencies and
characteristic parametrizations on spacetime coordinates $x^{i},$ metric $%
g_{ij},$ spin of particle, chosen types of spacetime connections etc}. For $%
\varphi =0,$ we get locally the standard mass/energy/momentum relation
describing a point particle in the special theory of relativity (SR).
Assuming $\ E\sim \frac{\partial }{\partial t},p_{i}\sim \frac{\partial }{%
\partial x^{i}}$ for some background bosonic media with ''effective light
velocity'' $c_{s}$ (see details in \cite{kost4,xiao,liberati}), the
nonlinear energy--momentum relation \ (\ref{defep}) results in
\begin{equation}
\omega ^{2}=c_{s}^{2}k^{2}+c_{s}^{2}\left( \frac{\overline{h}}{2m_{0}c_{s}}%
\right) ^{2}\text{ }k^{4}+...  \label{ndisp1}
\end{equation}%
when the ''phonon'' dispersion relation $\omega \approx c_{s}|k|$ violates
the acoustic Lorentz invariance with the wave length $\lambda =2\pi /|k|,$
for $k^{2}=(k_{1})^{2}+(k_{2})^{2}+(k_{3})^{2}$ and $\overline{h}=const.$ \
It is possible to derive more ''sophisticate'' dispersion relations with
cubic on $k$ and higher order terms and different coefficients than those in
(\ref{ndisp1}) if more general models of ''effective media'' with fermionic
and/or bosonic fields are considered.

\item \textbf{Finsler generating functions from nonlinear dispersions. }%
Nonlinear dispersions of type (\ref{ndisp1}) encode not only
''energy--momentum'' properties of point particles for LV. They contain also
a very fundamental information about possible metric elements defining more
general spacetime geometries than those postulated in SR and GR. Here, we
briefly present a setup for such constructions in terms of Finsler geometry %
\cite{lammer,stavr0,girelli,vacpr}. A Minkowski metric $\eta
_{ij}=diag[-1,+1,+1,+1]$ (for $i=1,2,3,4)$ defines a quadratic line element
in SR,%
\begin{equation}
ds^{2}=\eta
_{ij}dx^{i}dx^{j}=-(dx^{1})^{2}+(dx^{2})^{2}+(dx^{3})^{2}+(dx^{4})^{2},
\label{minkmetr}
\end{equation}%
with space type, $(x^{2},x^{3},x^{4}),$ and time like, $x^{1}=ct,$
coordinates where $c$ is the light speed.\footnote{%
Light rays can be parametrized as $x^{i}(\varsigma )$ with a real smooth
parameter $0\leq \varsigma \leq \varsigma _{0},$ when $ds^{2}/d\varsigma
^{2}=0;$ there is a ''null'' tangent vector field $y^{i}(\varsigma
)=dx^{i}/d\varsigma ,$ with $d\tau =dt/d\varsigma .$ Under general
coordinate transforms $x^{i^{\prime }}=x^{i^{\prime }}(x^{i}),$ we have $%
\eta _{ij}\rightarrow g_{i^{\prime }j^{\prime }}(x^{k});$ the condition $%
ds^{2}/d\varsigma ^{2}=0$ holds always for propagation of light, i.e. $%
g_{i^{\prime }j^{\prime }}y^{i^{\prime }}y^{j^{\prime }}=0.$} We can write
for some classes of coordinate systems (for simplicity, omitting priming of
indices and considering that indices of type $\widehat{i},\widehat{j}%
,...=2,3,4)$
\begin{equation}
c^{2}=g_{\widehat{i}\widehat{j}}(x^{i})y^{\widehat{i}}y^{\widehat{j}}/\tau
^{2}.  \label{lightf}
\end{equation}%
This formula can be used also in GR if we consider that $g_{\widehat{i}%
\widehat{j}}(x^{i})$ are solutions of Einstein equations. The above
quadratic on $y^{\widehat{i}}$ expression can be generalized to an arbitrary
nonlinear one, $\check{F}^{2}(y^{\widehat{j}}),$ in order to model
propagation of light in anisotropic media and/or \ for modeling an (ether)
spacetime geometry. We have to impose the the condition of homogeneity, $%
\check{F}(\beta y^{\widehat{j}})=\beta \check{F}(y^{\widehat{j}})$ for any $%
\beta >0,$ which is necessary for description of light propagation. The
formula (\ref{lightf}) transforms into
\begin{equation}
c^{2}=\check{F}^{2}(y^{\widehat{j}})/\tau ^{2}.  \label{lightfd}
\end{equation}%
Using approximations of type $\check{F}^{2}(y^{\widehat{j}})\approx \left(
\eta _{\widehat{i}\widehat{j}}y^{\widehat{i}}y^{\widehat{j}}\right) ^{r}+q_{%
\widehat{i}_{1}\widehat{i}_{2}...\widehat{i}_{2r}}y^{\widehat{i}_{1}}...y^{%
\widehat{i}_{2r}},$ for $r=1,2,....$ and $\widehat{i}_{1},\widehat{i}%
_{2},...,\widehat{i}_{2r}=2,3,4,$ we can parametrize small deformations of (%
\ref{lightf}) to (\ref{lightfd}). For $r=1$ and $q_{\widehat{i}_{1}\widehat{i%
}_{2}...\widehat{i}_{2r}}\rightarrow 0,$ we get the propagation of light
rays in SR. Instead of $\eta _{\widehat{i}\widehat{j}},$ we can introduce a
metric $g_{\widehat{i}\widehat{j}}(x^{i})$ from GR and include it in $\check{%
F}^{2}$ for gravitational fields when $\check{F}^{2}(x^{i},y^{\widehat{j}%
})\approx \left( g_{\widehat{i}\widehat{j}}(x^{k})y^{\widehat{i}}y^{\widehat{%
j}}\right) ^{r}+q_{\widehat{i}_{1}\widehat{i}_{2}...\widehat{i}%
_{2r}}(x^{k})y^{\widehat{i}_{1}}...y^{\widehat{i}_{2r}}$. For such
deformations (derived from (\ref{minkmetr}) and (\ref{lightf})), we get
generalized nonlinear homogeneous quadratic elements, {\small
\begin{equation}
ds^{2}=F^{2}(x^{i},y^{j})\approx -(cdt)^{2}+g_{\widehat{i}\widehat{j}%
}(x^{k})y^{\widehat{i}}y^{\widehat{j}}[1+\frac{1}{r}\frac{q_{\widehat{i}_{1}%
\widehat{i}_{2}...\widehat{i}_{2r}}(x^{k})y^{\widehat{i}_{1}}...y^{\widehat{i%
}_{2r}}}{\left( g_{\widehat{i}\widehat{j}}(x^{k})y^{\widehat{i}}y^{\widehat{j%
}}\right) ^{r}}] +O(q^{2}),  \label{fbm}
\end{equation}%
}when $F(x^{i},\beta y^{j})=\beta F(x^{i},y^{j}),$ for any $\beta >0.$ A
value $F$ is called a fundamental (generating) Finsler function usually
satisfying the condition that the Hessian
\begin{equation}
\ ^{F}g_{ij}(x^{i},y^{j})=\frac{1}{2}\frac{\partial F^{2}}{\partial
y^{i}\partial y^{j}}  \label{hess}
\end{equation}%
is not degenerate, see details in \cite{cartan,ma,bcs,vsgg}. For light rays,
the nonlinear element (\ref{fbm}) defines a nonlinear \ dispersion relation
between the frequency $\omega $ and the wave vector $k_{i},$\footnote{%
for simplicity, we can consider such a relation in a fixed point \ $%
x^{k}=x_{(0)}^{k},$ \ when $g_{\widehat{i}\widehat{j}}(x_{0}^{k})=g_{%
\widehat{i}\widehat{j}}$ and $q_{\widehat{i}_{1}\widehat{i}_{2}...\widehat{i}%
_{2r}}=q_{\widehat{i}_{1}\widehat{i}_{2}...\widehat{i}_{2r}}$ $(x_{0}^{k})$}%
\begin{equation}
\omega ^{2}=c^{2}[g_{\widehat{i}\widehat{j}}k^{\widehat{i}}k^{\widehat{j}%
}]^{2}\ (1-\frac{1}{r}\frac{q_{\widehat{i}_{1}\widehat{i}_{2}...\widehat{i}%
_{2r}}k^{\widehat{i}_{1}}...k^{\widehat{i}_{2r}}}{[g_{\widehat{i}\widehat{j}%
}k^{\widehat{i}}k^{\widehat{j}}]^{2r}}).  \label{disp}
\end{equation}
\end{enumerate}

The dispersion relations should be parametrized and computed differently for
various classes of theories formulated in terms of Finsler geometry and
generalizations. Here we cite a series of works on very special relativity %
\cite{gibbons,stavr1}, generalized (super) Finsler gravity and LV induced
from string gravity \cite{vstr1,vstr2,mavromatos}, double special relativity %
\cite{mignemi,ghosh}, Finsler--Higgs mechanism \cite{sindoni}, Finsler black
holes/ellipsoids induced by noncommutative variables \cite{vncfinsl}. In
particular, we can chose such subsets of coefficients $q_{\widehat{i}_{1}%
\widehat{i}_{2}...\widehat{i}_{2r}}$ when (\ref{disp}) transforms into (\ref%
{ndisp1}).

The main conclusion we derive from above points 1 and 2 is that various
classical and quantum gravity theories are with local nonlinear dispersions
of type (\ref{ndisp1}) and/or (\ref{disp}). Such theories are positively
with LV and can be characterized geometrically by nonlinear Finsler type
quadratic elements (\ref{fbm}) constructed as certain deformations of
standard quadratic elements for Minkowski (\ref{minkmetr}) and/or
pseudo--Riemannian spacetimes. This results in geometric constructions on
tangent, $TV$ (with local coordinated $u^{\alpha }=(x^{i},y^{a}),$ where $%
y^{a}$ label fiber coordinates; we shall write in brief $u=(x,y)$). We can
elaborate physical models on cotangent, $T^{\ast }V$ (with local coordinates
$\ \check{u}^{\alpha }=(x^{i},p_{a}),$ where $p_{a}$ label co--fiber
coordinates), bundles to a curved spacetime manifold $V$ (with local
coordinates $x^{i}=(x^{1},x^{2},x^{3},x^{4})$ of pseudo--Euclidean
signature). Constructions on $TM$ and $T^{\ast }V$ are typical for
Finsler--Lagrange, and/or Cartan--Hamilton geometries, and generalizations,
see details and references in \cite{cartan,ma,vsgg,vacpr}. In modern
particle physics and cosmology, see \cite{stavr2,lichang4,lichang6}, there
is a renewed interest in Finsler geometry applications, see reviews of
results and critical remarks in \cite{vcrit,vrevflg,vexsol,vqgr1,vqgr2}.

Nonlinear dispersions and associated Finsler like generating functions
suggest the idea that a self--consistent QG theory may be constructed not
just for a 4--d pseudo--Riemannian spacetime $V$ but for certain Finsler
type extensions on $TV$ and/or $T^{\ast }V.$ Following a nonholonomic
generalization of Fedosov deformation quantization, such quantum gravity
models were studied in \cite{vqgr1,vqgr2}. Roughly \ speaking, a QG model
with some generalized nonlinear dispersions, and associated fundamental
Finsler structures, should replace GR at very short distances approaching
the Planck length, $\mathit{l}_{P}\simeq \sqrt{\frac{\ ^{4}G\hbar }{c^{3}}}%
\simeq 1.6\times 10^{-33}cm,$ where $\ ^{4}G$ is the 4--d Newton constant
and $\hbar =h/2\pi $ is the Planck constant.

Over short distances, we have certain modifications of GR which seem to be
of Finsler type with additional depending ''velocity/momentum'' type
coordinates. A Finsler spacetime geometry/ gravity model is not completely
determined only by its nonlinear quadratic element $F(x,y)$ (\ref{fbm}), or
Hessian (\ref{hess}). It is completely stated after we choose (following
certain physical arguments) what types of metric tensor, $\ ^{F}\mathbf{g},$
nonlinear connection (N--connection), $^{F}\mathbf{N},$ and linear
connection, $^{F}\mathbf{D},$ are canonically induced by a generating
Finsler function $F(x,y)$ on $\widetilde{TV}\equiv TV/\{0\}$ (we exclude the
null sections $\{0\}$ over $TM).$\footnote{%
Following our conventions \cite{vcrit,vrevflg,vsgg}, we use ''boldface''
symbols for spaces/ geometrical/physical objects endowed with /adapted to
nonlinear connection structure, see definitions in next section; we also put
left up/low labels in order to emphasize that a geometric/physical object is
completely defined/induced by a corresponding fundamental generating
function, for instance, that $^{F}\mathbf{g}$ is completely and uniquely
determined by $F.$} The nature of QG and LV effects derived in certain
theoretical construction is related to a series of assumptions on
fundamental spacetime structure and considered classes of fundamental
equations, conservation laws and symmetries. For instance, it depends on the
fact if $^{F}\mathbf{D}$ and $\ ^{F}\mathbf{g}$ are compatible, or not; what
type of torsion $\ ^{F}\mathcal{T}$ of $\ ^{F}\mathbf{D}$ is induced by $F$
and/or $\ ^{F}\mathbf{g}$ and $\ ^{F}\mathbf{N,}$ if there are considered
compact and/or noncompact extra/velocity/momentum type dimensions etc (in
section \ref{s2}, we present rigorous definitions of such geometric/physical
objects).

In this work, our focus is on LV effects and QG phenomenology determined by
mechanisms for trapping/locallyzing gravitational and matter fields from a
Finsler spacetime on $TV$ to a 4--d observable pseudo--Riemannian spacetime $%
V.$ Such ideas were originally considered in brane gravity, see \cite%
{gogb1,gogb1a,gogb2} and \cite{rs1,rs2}, with a non--compact extra dimension
coordinate. In Finsler gravity theories on tangent bundles of 3--d and/or
4--d pseudo--Riemannian spacetimes, there are considered respectively 3+3
and/or 4+4 dimensional locally anisotropic spacetime models determined by
data $\left[ F:\ ^{F}\mathbf{g,\ }^{F}\mathbf{N,\ }^{F}\mathbf{D}\right] $.
This contains a more rich geometric structure than that for Einstein spaces
determined by a metric $\mathbf{g,}$ and its unique torsionless and metric
compatible Levi--Civita connection $\ ^{g}\nabla .$

Warped Finsler like configurations and related trapping ''isotropization''
have to be adapted to the N--connection structure (as we studied for the
case of locally anisotropic black holes and propagating solitonically black
holes and wormholes \cite{vsing1}). We can consider solutions with both
types of exponential and non--exponential factors by introducing
non--gravitational interactions or considering a pure gravitational trapping
mechanism for all types of spin fields similarly to 5-d and 6--d
pseudo--Riemannian configurations in \cite{midod,gm2,gs2,singlbr}. The
physics of locally isotropic brane theories with extra dimensions and the
Finsler--brane models are very different even the LV effects (see \cite%
{gensol2,vexsol,grojean} for locally isotropic branes in 5--d) can be
computed in both cases\footnote{%
the Finsler configurations are or different nature when the generating
functions are determined by certain general coefficients in a QG model, and
related LV, which is not the case, for instance, of Rundall--Sundrum branes}.

The portion of this paper developing conceptual and theoretical issues of
Finsler gravity and brane theories spans sections \ref{s2} and \ref{s3}. It
begins in section \ref{s2} with a review of the Einstein--Finsler gravity
model and the anholonomic deformation method of constructing exact solutions
for gravitational field equations. Section \ref{s3} concerns explicit
Finsler--brane solutions in 3+3 and 4+4 dimensional gravity on tangent
bundles. Finally, in section \ref{s4} we conclude the results. For
convenience, in Appendices we provide two important Theorems and relevant
computations on constructing exact solutions in Einstein and Finsler gravity
theories. Throughout the paper, we follow the conventions of Refs. \cite%
{vrevflg,vsgg} and \cite{vacpr} where possible.

\section{Einstein--Finsler Gravity}

\label{s2} In general, there are two different classes of Finsler gravity
models which can be constructed on a $\mathbf{TV}=(TV,\pi ,V),$ where $TV$ \
is the total space, $\pi $ is a surjective projection and $V,\dim V=n$. In
this work, we consider that $V$ is a pseudo--Riemannian/Einstein manifold of
dimension $n=2,3,$ or 4), is the base manifold, see details and critical
remarks in Refs \cite{vcrit,vacpr,vrevflg,vsgg}.

The first class of theories is with metric noncompatible linear Finsler
connections $\ ^{F}\mathbf{D}$ when $\ ^{F}\mathbf{D}\ ^{F}\mathbf{g=\
^{F}Q\neq 0}$ (the typical example is that when $\ ^{F}\mathbf{D=}\ ^{Ch}%
\mathbf{D}$ is the Chern connection for which the total $\ ^{Ch}\mathbf{Q}$
is not zero but torsion vanishes, $\ ^{Ch}\mathcal{T}=0).$ Because of
nonmetricity, it seems that there are a number of conceptual/theoretical and
technical problems with definition of spinors and Dirac operators,
conservations laws and performing quantization of such theories \cite%
{vcrit,vacpr,vrevflg}. In our opinion, such geometries have less
perspectives for applications in standard particle physics and ''simple''
modifications, for instance, for purposes of modern cosmology.

The second class of Finsler gravity models is with such $\ ^{F}\mathbf{D}$
which are metric compatible, i.e. $\ ^{F}\mathbf{D}\ ^{F}\mathbf{g}=0.$ Such
a locally anisotropic gravity theory is positively with nontrivial torsion, $%
\ ^{F}\mathcal{T}\neq 0.$ A very important property is that there are $\ ^{F}%
\mathbf{D}$ when $\ ^{F}\mathcal{T}$ \ is completely defined by the total
Finsler metric structure $\ ^{F}\mathbf{g}$ and a prescribed nonlinear
connection (N--connection) $\ ^{F}\mathbf{N}$. For instance, this is the
case of canonical distinguished connection (d--connection) $\ ^{F}\mathbf{D=%
\widehat{\mathbf{D}},}$ see details in \cite{ma,vrevflg,vsgg}, and the
Cartan d--connection $\ ^{F}\mathbf{D}=\widetilde{\mathbf{D}},$ see formula (%
\ref{cdc}) below. There are preferred constructions with $\widetilde{\mathbf{%
D}}$ because it defines canonically an almost K\"{a}hler structure (for
instance, this is important for deformation/ A--brane quantization of
gravity \cite{vqgr1,vqgr2}).

An an \textbf{Einstein--Finsler gravity} theory (EFG), we consider a model
of gravity on $TV$ defined by data $\left[ F:\ ^{F}\mathbf{g,\ }^{F}\mathbf{%
N,\ }^{F}\mathbf{D=}\widetilde{\mathbf{D}}\right] $ and corresponding
gravitational field equations in such variables (see section \ref{ss2})
following the same principles (postulates) as in GR stated by data $\left[
\mathbf{g,}\ ^{g}\nabla \right].$ Additionally, we suppose that there is a
trapping/warped mechanism defined by explicit solutions of (Finsler type)
gravitational field equations which in classical limits for $\mathit{l}%
_{P}\rightarrow 0,$ when EFG $\rightarrow $ GR, determining QG corrections
to gravitational and matter field interactions at different scales depending
on the class of considered models and solutions.

\subsection{Fundamental objects in EFG}

A (pseudo) \textsf{Finsler space} $F^{n}=(V,F)$ corresponding, for instance,
to a (pseudo) Riemannian manifold $V$ of signature $(-,+,+,...)$ consists of
a Finsler metric (fundamental/generating function) $F(x,y)$ (\ref{fbm})
defined as a real valued function $F:TV\rightarrow \mathbb{R}$ with the
properties that the restriction of $F$ to $\widetilde{TM}$ is a function 1)
positive; 2) of class $C^{\infty }$ and $F$ is only continuous on $\{0\};$
3) positively homogeneous of degree 1 with respect to $y^{i},$ i.e. $%
F(x,\beta y)=|\beta |F(x,y),$\ $\beta \in \mathbb{R};$ and 4) the Hessian $\
^{F}g_{ij}=(1/2)\partial ^{2}F/\partial y^{i}\partial y^{j}$ (\ref{hess})
defined on $\widetilde{TV},$ is nondegenerate (for Finsler spaces, this
condition is changed into that $\ ^{F}g_{ij}$ is positive definite. In
brief, a Finsler space is a Lagrange space with effective Lagrangian $%
L=F^{2}.$\footnote{%
Similar theories can be elaborated for (pseudo) Lagrange spaces and
generalizations as it is provided in \cite{ma,vrevflg,vsgg}. This way, we
can construct different ''analogous gravity'' and geometric mechanics
models. Here we also note that Finsler--Lagrange variables can be introduced
even in Einstein gravity which is very convenient for constructing exact
solutions and developing certain \ models of QG.}

\subsubsection{The canonical (Finsler) N--connection}

One of three fundamental geometric objects induced by a Finsler metric $F,$
and defining a Finsler space, is the nonlinear connection (N--connection). A
N--connection $\mathbf{N}$ is by definition ($:=$) a Whitney sum
\begin{equation}
TTV:=hTV\oplus vTV.  \label{whitney}
\end{equation}%
Geometrically, this an example of nonholonomic (equivalentl, anholonomic, or
non--integrable) distribution with conventional horizontal (h) -- vertical
(v) decomposition/ splitting which can be considered for the module of
vector fields $\chi (TTV)$ on $TV.$ For instance, $\mathbf{Y}=\ ^{h}\mathbf{Y%
}+\ ^{v}\mathbf{Y}$ for any vector $\mathbf{Y}\in \chi (TTV),$ where $\ ^{h}%
\mathbf{Y}\doteqdot h\mathbf{Y}\in \chi (hTV)$ and $\ ^{v}\mathbf{Y}%
\doteqdot v\mathbf{Y}\in \chi (vTV).$

There is a canonical N--connection structure $\mathbf{N=\ }^{c}\mathbf{N}$ $%
\ $which is defined by $F$ following such arguments. Considering that $%
L=F^{2}$ is a regular Lagrangian (i.e. with nondegenerate $^{F}g_{ij}$ (\ref%
{hess})) and define the action integral $S(\tau
)=\int\limits_{0}^{1}L(x(\tau ),y(\tau ))d\tau $, with $y^{k}(\tau
)=dx^{k}(\tau )/d\tau ,$ for $x(\tau )$ parametrizing smooth curves on $V$ \
with $\tau \in \lbrack 0,1]$. \ By straightforward computations, we can
prove that the Euler--Lagrange equations of $S(\tau ),$ i.e. $\frac{d}{d\tau
}\frac{\partial L}{\partial y^{i}}-\frac{\partial L}{\partial x^{i}}=0,$ are
equivalent to the ''nonlinear geodesic'' (equivalently, semi--spray)
equations $\frac{d^{2}x^{k}}{d\tau ^{2}}+2G^{k}(x,y)=0,$ where $G^{k}=\frac{1%
}{4}g^{kj}\left( y^{i}\frac{\partial ^{2}L}{\partial y^{j}\partial x^{i}}-%
\frac{\partial L}{\partial x^{j}}\right) $ defines the canonical
N--connection $\mathbf{\ }^{c}\mathbf{N=}\{\ ^{c}N_{j}^{a}\},$ where $\
^{c}N_{j}^{a}=\frac{\partial G^{a}(x,y)}{\partial y^{j}}.$

Under general (co) frame/coordinate transform, $\mathbf{e}^{\alpha
}\rightarrow \mathbf{e}^{\alpha ^{\prime }}=e_{\ \alpha }^{\alpha ^{\prime }}%
\mathbf{e}^{\alpha }$ and/or $u^{\alpha }\rightarrow u^{\alpha ^{\prime
}}=u^{\alpha ^{\prime }}(u^{\alpha }),$ preserving the splitting (\ref%
{whitney}), we transform $\ ^{c}N_{j}^{a}\rightarrow N_{j^{\prime
}}^{a^{\prime }},$ when $\mathbf{N}=N_{i^{\prime }}^{a^{\prime
}}(u)dx^{i^{\prime }}\otimes \frac{\partial }{\partial y^{a^{\prime }}}$ is
given locally by a set of coefficients $\{N_{j}^{a}\}$. Hereafter, we shall
omit priming, underlying etc of indices if that will not result in
ambiguities.

\subsubsection{Sasaki types lifts of metrics in Finsler spaces}

For a fundamental Finsler function $F(x,y),$ we can construct a canonical
(Sasaki type) metric structure
\begin{eqnarray}
\ ^{F}\mathbf{g} &=&\ ^{F}g_{ij}(x,y)\ e^{i}\otimes e^{j}+\left( \mathit{l}%
_{P}\right) ^{2}\ ^{F}g_{ij}(x,y)\ \ ^{F}\mathbf{e}^{i}\otimes \ \ ^{F}%
\mathbf{e}^{j},  \label{slm} \\
e^{i} &=&dx^{i}\mbox{ and }\ \ ^{F}\mathbf{e}^{a}=dy^{a}+\
^{F}N_{i}^{a}(u)dx^{i},  \label{ddifl}
\end{eqnarray}%
where $\ ^{F}\mathbf{e}^{\mu }=(e^{i},\ ^{F}\mathbf{e}^{a})$ (\ref{ddifl})
is the dual to $\ ^{F}\mathbf{e}_{\alpha }=(\ ^{F}\mathbf{e}_{i},e_{a}),$
for
\begin{equation}
\ ^{F}\mathbf{e}_{i}=\frac{\partial }{\partial x^{i}}-\ ^{F}N_{i}^{a}(u)%
\frac{\partial }{\partial y^{a}}\mbox{ and }e_{a}=\frac{\partial }{\partial
y^{a}}.  \label{dder1}
\end{equation}%
We shall put the square of an effective Planck length $\mathit{l}_{P}$
before the v--part of metric (\ref{slm}) if we shall wont to have the same
dimensions for the h-- and v--components of metric when coordinates have the
dimensions $[x^{i}]=cm$ and $[y^{i}\sim dx^{i}/ds]=cm/cm.$

Using frame transforms $e^{\alpha ^{\prime }}=e_{\ \alpha }^{\alpha ^{\prime
}}\mathbf{e}^{\alpha },$ any metric
\begin{equation}
\mathbf{g}=g_{\alpha \beta }du^{\alpha }\otimes du^{\beta }  \label{cm}
\end{equation}%
on $TM,$\footnote{%
a dual local coordinate basis os $du^{\beta }=(dx^{j},dy^{b}),$ when $%
\partial _{\alpha }=\partial /\partial u^{\alpha }=(\partial _{i}=\partial
/\partial x^{i},\partial _{a}=\partial /\partial y^{a})$} including $\ ^{F}%
\mathbf{g}$ (\ref{slm}), can be represented in N--adapted form
\begin{equation}
\ \mathbf{g}=\ g_{ij}(x,y)\ e^{i}\otimes e^{j}+\left( \mathit{l}_{P}\right)
^{2}\ h_{ab}(x,y)\ \mathbf{e}^{a}\otimes \ \mathbf{e}^{b},  \label{dm}
\end{equation}%
for an N--adapted base \ $\mathbf{e}_{\nu }=(\mathbf{e}_{i},e_{a}),$ where
\begin{equation}
\mathbf{e}_{i}=\frac{\partial }{\partial x^{i}}-\ N_{i}^{a}(u)\frac{\partial
}{\partial y^{a}}\mbox{ and
}e_{a}=\frac{\partial }{\partial y^{a}},  \label{nader}
\end{equation}%
and the dual frame (coframe) structure is $\mathbf{e}^{\mu }=(e^{i},\mathbf{e%
}^{a}),$ for
\begin{equation}
e^{i}=dx^{i}\mbox{ and }\mathbf{e}^{a}=dy^{a}+\ N_{i}^{a}(u)dx^{i}.
\label{nadif}
\end{equation}

The local bases induced by N--connection structure, for instance, (\ref%
{nader}) satisfy nontrivial nonholonomy relations of type
\begin{equation}
\lbrack \mathbf{e}_{\alpha },\mathbf{e}_{\beta }]=\mathbf{e}_{\alpha }%
\mathbf{e}_{\beta }-\mathbf{e}_{\beta }\mathbf{e}_{\alpha }=W_{\alpha \beta
}^{\gamma }\mathbf{e}_{\gamma },  \label{anhrel}
\end{equation}%
with (antisymmetric) nontrivial anholonomy coefficients $W_{ia}^{b}=\partial
_{a}N_{i}^{b}$ and $W_{ji}^{a}=\Omega _{ij}^{a}$ determined by the
coefficients of curvature of N--connection.

The above formulas define h-- and v--splitting of metrics on $TM,$
respectively, $\ ^{h}\mathbf{g=}\{g_{ij}(u)\}$ and $\ ^{v}\mathbf{g=}%
\{h_{ab}(u)\}.$ Extending the principle of general covariance from $V$ to $%
TV,$ i.e. from GR to EFG, we can work equivalently with any parametrization
of metrics in the form (\ref{slm}), (\ref{cm}), or (\ref{dm}). The first
parametrization show in explicit form that our gravity model is for a
Finsler spacetime, the second one states the coefficients of metric with
respect to local coordinate (co) bases and the third one will be convenient
for constructing exact solutions in EFG.

\subsubsection{Canonical linear/distinguished connections}

For any Finsler metric $\ ^{F}\mathbf{g}$ (\ref{slm}), we can compute in
standard form the Levi--Civita connection $\ ^{F}\nabla .$ But such a linear
connection is not used in Finsler geometry because it is not adapted to the
N--connection structure. We have to revise the concept of linear connection
for nonholonomic bundles/manifolds enabled with splitting of type (\ref%
{whitney}):\ A distinguished connection (d--connection) is a linear
connection $\mathbf{D}$ \ preserving by parallelism the N--connection
splitting (\ref{whitney}).\footnote{%
In Lagrange--Finsler geometry, there are used the terms distinguished
tensor/ vector / spinor connection \ etc (d--tensor, d--vector, d--spinor,
d--connection etc) \cite{ma,vrevflg,vsgg} \ for the corresponding geometric
objects defined with respect to N--adapted (co) bases.}

To a d--connection $\mathbf{D}=(\ ^{h}D,\ ^{v}D)=(L_{\ jk}^{i},C_{jc}^{i}), $
for $L_{\ jk}^{i}=L_{\ bk}^{a}$ and $C_{jc}^{i}=C_{bc}^{a}$ (with a chosen
contraction for h- and v--indices), we can associate a 1--form $\mathbf{%
\Gamma }_{\ \beta }^{\alpha }=[\mathbf{\Gamma }_{\ j}^{i},\mathbf{\Gamma }%
_{\ b}^{a}]$ with
\begin{equation*}
\mathbf{\Gamma }_{\ j}^{i}=\mathbf{\Gamma }_{\ j\gamma }^{i}\mathbf{e}%
^{\gamma }=L_{\ jk}^{i}e^{k}+C_{jc}^{i}\mathbf{e}^{c},\ \mathbf{\Gamma }_{\
b}^{a}=\mathbf{\Gamma }_{\ b\gamma }^{a}\mathbf{e}^{\gamma }=L_{\
bk}^{a}e^{k}+C_{bc}^{a}\mathbf{e}^{c}.
\end{equation*}%
The torsion, $\mathcal{T}$ $=\{\mathbf{T}_{\beta \gamma }^{\alpha }\},$\ and
curvature, $\mathcal{R}=\{\mathbf{R}_{\ \beta \gamma \tau }^{\alpha }\},$
tensors\ of a d--connection $\mathbf{D}$ are defined\ and computed in usual
forms as for linear connections for any $\mathbf{X,Y,Z}\in \chi (TTV)$.
Using Cartan's structure equations%
\begin{eqnarray}
de^{i}-e^{k}\wedge \mathbf{\Gamma }_{\ k}^{i} &=&-\mathcal{T}^{i},\ d\mathbf{%
e}^{a}-\mathbf{e}^{b}\wedge \mathbf{\Gamma }_{\ b}^{a}=-\mathcal{T}^{a},
\notag \\
d\mathbf{\Gamma }_{\ j}^{i}-\mathbf{\Gamma }_{\ j}^{k}\wedge \mathbf{\Gamma }%
_{\ k}^{i} &=&-\mathcal{R}_{j}^{i},  \label{cartseq}
\end{eqnarray}%
we can compute the N--adapted coefficients of torsion and curvature, see
details in \cite{ma,vrevflg,vsgg}. For instance, an explicit computation
results in
\begin{equation}
\mathcal{T}^{i}=C_{\ jc}^{i}e^{i}\wedge \mathbf{e}^{c}\mbox{
and }\mathcal{T}^{a}=-\frac{1}{2}\Omega _{ij}^{a}e^{i}\wedge e^{j}+\left(
e_{b}N_{i}^{a}-L_{\ bi}^{a}\right) e^{i}\wedge \mathbf{e}^{b},
\label{nztors}
\end{equation}%
with nontrivial values (anti--symmetric on lower indices) of $T_{\
jc}^{i}=-T_{\ cj}^{i}=C_{\ jc}^{i},$ $T_{\ ji}^{a}=-T_{\ ij}^{a}=\frac{1}{2}%
\Omega _{ij}^{a},$ $T_{\ bi}^{a}=-T_{\ ib}^{a}=e_{b}N_{i}^{a}-L_{\ bi}^{a}.$

For a metric structure $\mathbf{g}=[g_{ij},h_{ab}]$ (\ref{dm}), there is a
unique normal d--connection $\widetilde{\mathbf{D}}$ \ which is metric
compatible, $\widetilde{\mathbf{D}}\ \mathbf{g=0,}$ and with vanishing $hhh$%
- and $vvv$--components ($\ ^{h}\widetilde{\mathcal{T}}(h\mathbf{X,}h\mathbf{%
Y})=0$ and $\ \ ^{v}\widetilde{\mathcal{T}}(v\mathbf{X,}v\mathbf{Y})=0,$ for
any vectors $\mathbf{X}$ and $\mathbf{Y)}$ of torsion$\ \widetilde{\mathcal{T%
}}$\ computed following formulas (\ref{nztors}). If $\mathbf{g}=\ ^{F}%
\mathbf{g,}$ we get get the coefficients of the so--called Cartan
d--connection in Finsler geometry \cite{cartan,ma,vrevflg,vsgg}. We can
verify that locally the normal d--connection $\widetilde{\mathbf{D}}=(\ ^{h}%
\widetilde{\mathbf{D}}\mathbf{,}\ ^{v}\widetilde{\mathbf{D}})$ is given
respectively by coefficients $\widetilde{\mathbf{\Gamma }}_{\ \beta \gamma
}^{\alpha }=\left( \widetilde{L}_{\ bk}^{a},\widetilde{C}_{bc}^{a}\right) ,$
\begin{equation}
\widetilde{L}_{\ jk}^{i}=\frac{1}{2}g^{ih}(\mathbf{e}_{k}g_{jh}+\mathbf{e}%
_{j}g_{kh}-\mathbf{e}_{h}g_{jk}),\widetilde{C}_{\ bc}^{a}=\frac{1}{2}%
h^{ae}(e_{b}h_{ec}+e_{c}h_{eb}-e_{e}h_{bc}),  \label{cdc}
\end{equation}%
are computed with respect to N--adapted frames. The covariant h--derivative
is $\ ^{h}\widetilde{\mathbf{D}}=\{\widetilde{L}_{\ jk}^{i}\}$ and
v--derivative is $\ ^{v}\widetilde{\mathbf{D}}=\{\widetilde{C}_{bc}^{a}\}.$
The torsion coefficients $\widetilde{\mathbf{T}}_{\ \beta \gamma }^{\alpha }$
of $\widetilde{\mathbf{D}}$ are $\widetilde{T}_{jk}^{i}=0$ and $\widetilde{T}%
_{bc}^{a}=0$ but with non--zero cross coefficients, $\widetilde{T}%
_{ij}^{a}=\Omega _{ij}^{a},\widetilde{T}_{ib}^{a}=e_{b}N_{i}^{a}-\widetilde{L%
}_{\ bi}^{a}.$

\subsubsection{Finsler variables in (pseudo) Riemannian geometry}

Finsler variables can be introduced not only on $TM$ but also, via
corresponding nonholonomic distributions, on any pseudo--Riemannian manifold %
\cite{gensol2,vexsol,vqgr2,vrevflg} $\mathbf{V},\dim \mathbf{V}=2n,n\geq 2,$
enabled with metric structure $\mathbf{g}$. On such a manifold, we can
prescribe any type of nonholonomic frames/ distributions. For instance, we
can choose a distribution defined by a regular generating function of
necessary type homogeneity, $F(x,y),$ when coordinates $u=(x,y)$ are local
ones on $\mathbf{V,}$ with nondegenerate Hessian $\ ^{F}g_{ij}$, and define $%
\mathbf{g}=\ ^{F}\mathbf{g}$. We model on $\mathbf{V}$ a Finsler geometry if
we construct from $\ ^{F}\mathbf{g,}$ in a unique form, the Cartan
d--connection $\ \widetilde{\mathbf{D}}.$

In ''standard'' variables, a (pseudo) Riemannian geometry is characterized
by the Levi--Civita connection $\nabla .$ \footnote{%
By definition, it is metric compatible, $\nabla \mathbf{g}=0,$ and
torsionless, $\ ^{\nabla }\mathcal{T}=0.$} We have
\begin{equation}
\mathbf{\ }\widetilde{\mathbf{D}}=\mathbf{\ }^{F}\nabla +\mathbf{\ }%
\widetilde{\mathbf{Z}},  \label{dist}
\end{equation}%
where the distortion tensor $\mathbf{\ }^{F}\widetilde{\mathbf{Z}}$ \ is
determined by the torsion $\widetilde{\mathcal{T}},$ see explicit
coefficients (\ref{nztors}). All such geometric objects (i.e.$\ \widetilde{%
\mathbf{D}},\ ^{F}\nabla , \widetilde{\mathbf{Z}})$ are completely defined
by the same metric structure $\mathbf{g.}$ Any geometric (pseudo) Riemannian
data $(\mathbf{g,\nabla })$ can be transformed equivalently into $(\mathbf{g}%
=\ ^{F}\mathbf{g,}\widetilde{\mathbf{D}})$ and inversely.

The question of (at least formal) equivalence of two gravity theories given
by data/ variables $[F:\ ^{F}\mathbf{g,\ }^{F}\mathbf{N,\ }^{F}\mathbf{D}=%
\widetilde{\mathbf{D}}]$ or $[\mathbf{g=\ ^{F}\mathbf{g},}\ ^{g}\nabla =%
\mathbf{\ }^{F}\nabla = \widetilde{\mathbf{D}}-\widetilde{\mathbf{Z}}]$ (on $%
TV,$ or $\mathbf{V)}$ depends on the type of gravitational field equations
(for $\mathbf{\ }\widetilde{\mathbf{D}}$ or $\nabla $) and matter field
sources are postulated for a model of relativity theory.

\subsection{Field equations in EFG}

\label{ss2} We can elaborate a Finsler gravity theory on $TM$ using the
d--connection $\widetilde{\mathbf{D}}$ and following in general lines the
same postulates as in GR. Such a model present a minimal metric compatible
Finsler extension of the Einstein gravity but for the generating function $F$%
.

The curvature 2--form of $\widetilde{\mathbf{D}}=\{\widetilde{\mathbf{\Gamma
}}_{\beta \gamma }^{\alpha }\}$ is computed (see (\ref{cartseq}))%
\begin{equation*}
\widetilde{\mathcal{R}}_{\ \gamma }^{\tau }=\widetilde{\mathbf{R}}_{\ \gamma
\alpha \beta }^{\tau }\ \mathbf{e}^{\alpha }\wedge \ \mathbf{e}^{\beta }=%
\frac{1}{2}\widetilde{R}_{\ jkh}^{i}e^{k}\wedge e^{h}+\widetilde{P}_{\
jka}^{i}e^{k}\wedge \mathbf{e}^{a}+\frac{1}{2}\widetilde{S}_{\ jcd}^{i}%
\mathbf{e}^{c}\wedge \mathbf{e}^{d},
\end{equation*}%
when the nontrivial N--adapted coefficients of curvature $\ \widetilde{%
\mathbf{R}}_{\ \beta \gamma \tau }^{\alpha }$ are
\begin{eqnarray*}
\widetilde{R}_{\ hjk}^{i} &=&\mathbf{e}_{k}\widetilde{L}_{\ hj}^{i}-\mathbf{e%
}_{j}\widetilde{L}_{\ hk}^{i}+\widetilde{L}_{\ hj}^{m}\widetilde{L}_{\
mk}^{i}-\widetilde{L}_{\ hk}^{m}\widetilde{L}_{\ mj}^{i}-\widetilde{C}_{\
ha}^{i}\Omega _{\ kj}^{a},  \label{ncurv} \\
\widetilde{P}_{\ jka}^{i} &=&e_{a}\widetilde{L}_{\ jk}^{i}-\widetilde{%
\mathbf{D}}_{k}\widetilde{C}_{\ ja}^{i},\ \widetilde{S}_{\ bcd}^{a}=e_{d}%
\widetilde{C}_{\ bc}^{a}-e_{c}\widetilde{C}_{\ bd}^{a}+\widetilde{C}_{\
bc}^{e}\widetilde{C}_{\ ed}^{a}-\widetilde{C}_{\ bd}^{e}\widetilde{C}_{\
ec}^{a}.  \notag
\end{eqnarray*}

The Ricci tensor $\widetilde{R}ic=\{\widetilde{\mathbf{R}}_{\alpha \beta }\}$
is defined by contracting respectively the components of curvature tensor, $%
\widetilde{\mathbf{R}}_{\alpha \beta }\doteqdot \widetilde{\mathbf{R}}_{\
\alpha \beta \tau }^{\tau }.$ The scalar curvature is $\ ^{s}\widetilde{%
\mathbf{R}}\doteqdot \mathbf{g}^{\alpha \beta }\widetilde{\mathbf{R}}%
_{\alpha \beta }=g^{ij}\widetilde{R}_{ij}+h^{ab}\widetilde{R}_{ab},$ where $%
\widetilde{R}=g^{ij}\widetilde{R}_{ij}$ and $\ \widetilde{S}=h^{ab}%
\widetilde{R}_{ab}$ are respectively the h-- and v--components of scalar
curvature.

The gravitational field equations for our Finsler gravity model with metric
compatible d--connection$\ ^{F}\mathbf{D=}\widetilde{\mathbf{\mathbf{D}}},$
\begin{equation}
\widetilde{\mathbf{E}}_{\ \beta \delta }=\widetilde{\mathbf{R}}_{\ \beta
\delta }-\frac{1}{2}\mathbf{g}_{\beta \delta }\ ^{s}\widetilde{R}=\widetilde{%
\mathbf{\Upsilon }}_{\beta \delta }  \label{ensteqcdc}
\end{equation}
can be introduced in geometric and/or variational forms on $TM,$ similarly
to Einstein equations in GR,
\begin{equation}
\ _{\shortmid }E_{\ \beta \delta }=\ _{\shortmid }R_{\ \beta \delta }-\frac{1%
}{2}\mathbf{g}_{\beta \delta }\ \ _{\shortmid }^{s}R=\ _{\shortmid }\Upsilon
_{\beta \delta },  \label{einsteqlc}
\end{equation}%
where all values (the Einstein and Ricci \ tensors, respectively, $\
_{\shortmid }E_{\ \beta \delta }$ and $\ _{\shortmid }R_{\ \beta \delta }$,
scalar curvature, $\ _{\shortmid }^{s}R,$ and the energy--momentum tensor, $%
\ _{\shortmid }\Upsilon _{\beta \delta }$) are for the Levi--Civita
connection $\ ^{F}\nabla $ computed for the same $\mathbf{g}_{\beta \delta
}=\ ^{F}\mathbf{g}_{\beta \delta }.$

A source $\widetilde{\mathbf{\Upsilon }}_{\beta \delta }$ can be defined
following certain geometric and/or N--adapted variational principles for
matter fields, see such examples in \cite{vsgg}. An important property of
the equations (\ref{ensteqcdc}) is that it can be integrated in very general
forms. On exact solutions for such equations (related to black hole physics,
locally anisotropic thermodynamics etc) see \cite%
{gensol2,vexsol,vqgr2,vrevflg,vsgg} and references therein. Finsler modified
Einstein equations of type (\ref{ensteqcdc}) can be such way constructed
that they would be equivalent to the Einstein equations for $\nabla$.%
\footnote{%
This is possible for the normal/ Cartan d--connection $\widetilde{\mathbf{D}}
$ being completely defined by $\mathbf{g}_{\beta \delta }$ (\ref{dm}) and
if\ $\ \widetilde{\mathbf{\Upsilon }}_{\beta \delta }=\ ^{matter}\mathbf{%
\Upsilon }_{\beta \delta }+\ ^{z}\mathbf{\Upsilon }_{\beta \delta }$ are
derived in such a way that they contain contributions from\ 1) \ the
N--adapted energy--momentum tensor (defined variationally following the same
principles as in GR but on $TV$) and 2) the distortion of the Einstein
tensor in terms of $\ \widehat{\mathbf{Z}}$ (\ref{dist}), $\widetilde{%
\mathbf{Z}}_{\ \beta \delta }=\ _{\shortmid }E_{\alpha \beta }+\ ^{z}%
\widetilde{\mathbf{Z}}_{\ \beta \delta },$ for $\ ^{z}\mathbf{\ }\widetilde{%
\mathbf{Z}}_{\ \beta \delta }=\ ^{z}\mathbf{\Upsilon }_{\beta \delta }.$ The
value $\ ^{z}\widetilde{\mathbf{Z}}_{\ \beta \delta }$ is computed by
introducing $\widetilde{\mathbf{D}}=\ ^{F}\nabla +\mathbf{\ }\widetilde{%
\mathbf{Z}}$ into (\ref{ncurv}) and corresponding contractions of indices in
order to find the Ricci d--tensor and scalar curvature.} Such an equivalence
is important if we reformulate the GR theory in Finsler, or almost K\"{a}%
hler variables \cite{gensol2,vexsol,vqgr2}, but there are not strong
theoretical and/or experimental arguments to impose such conditions for
Finsler gravity theories on $TM.$

Finally, we emphasize that the EFG theory is positively with nontrivial
torsion structure $\widetilde{\mathbf{T}}_{\ \beta \gamma }^{\alpha }$
induced by fundamental generating function $F(x,y).$ This torsion is
completely defined by certain off--diagonal coefficients of the metric
structure $\ ^{F}\mathbf{g,}$ including $\ ^{F}\mathbf{N.}$

\section{Finsler--Branes}

\label{s3} Examples of Einstein--Finsler gravity model and QG phenomenology
can be elaborated for metrics $\ ^{F}\mathbf{g},$ see parametrizations (\ref%
{slm}) \ and (\ref{dm}), transforming into Einstein metrics for $\mathit{l}%
_{P}\rightarrow 0.$ In the classical limit, the gravitational physics is
satisfactory described by GR (perhaps with certain exceptions related to
accelerating Universes and dark energy/matter problems (see \cite%
{stavr2,vsgg,lichang4,lichang6,vcrit}). In this section, we study scenarios
of QG phenomenology and LV when classical 4--d Einstein spacetimes are
embedded into 8--d Finsler spaces with non--factorizable velocity type
coordinates. Experimentally, the light velocity is finite and metrics in GR
do not depend explicitly on velocity/momentum type variables which can be
modelled via trapping/warping solutions in EFG.

\subsection{General ansatz and integrable filed equations}

The system of equations (\ref{ensteqcdc}) can be integrated in very general
forms (following geometric methods reviewed in details in Refs. \cite%
{gensol2,vexsol}). In this paper, we can use a simplified approach because
our 8--d Finsler gravity models are with Killing symmetries and Finsler
branes can be described by some off--diagonal ansatz for metrics and
connections.\footnote{%
For convenience, we provide in Appendix two theorems on constructing exact
solutions for a 4--d Einstein--Finsler toy model which is exactly
integrable. Various extensions of the outlined there anholonomic deformation
method to 6--d and 8--d Finsler brane spacetimes with nontrivial
N--connection structures are straightforward.}

It possible to extend GR theory to holonomic 8--d models on tangent bundle
considering a trivial N--connection/Finsler structure for the EFG when
solutions with diagonal metrics play an important role. To select a more
realistic model of velocity/momentum depending gravity, we have to solve the
8--d Einstein equations (\ref{einsteqlc}) (defining a "velocity" depending
type of scalar--tensor gravity theory, see discussion in Ref. \cite{vacpr})
and compare such classes of solutions with generic off--diagonal ones and
nontrivial d--torsion and N--connection structures constructed for Finsler
gravity.

We use an ansatz which via frame transform can be parametrized {\small
\begin{eqnarray}
&&\mathbf{g} =\ \phi ^{2}(y^{5})[g_{1}(x^{k})\ e^{1}\otimes
e^{1}+g_{2}(x^{k})\ e^{2}\otimes e^{2} +h_{3}(x^{k},v)\ \mathbf{e}%
^{3}\otimes \mathbf{e}^{3}+  \notag \\
&& h_{4}(x^{k},v)\ \mathbf{e}^{4}\otimes \mathbf{e}^{4}] +\left( \mathit{l}%
_{P}\right) ^{2}\ [h_{5}(x^{k},v,y^{5})\ \mathbf{e}^{5}\otimes \ \mathbf{e}%
^{5}+h_{6}(x^{k},v,y^{5})\ \mathbf{e}^{6}\otimes \ \mathbf{e}^{6}]  \notag \\
&&+\left( \mathit{l}_{P}\right) ^{2}\ [h_{7}(x^{k},v,y^{5},y^{7})\ \mathbf{e}%
^{7}\otimes \ \mathbf{e}^{7}+h_{8}(x^{k},v,y^{5},y^{7})\ \mathbf{e}%
^{8}\otimes \ \mathbf{e}^{8}],  \label{ans8d} \\
&& \mathbf{e}^{3} =dv+w_{i}dx^{i},\ \mathbf{e}^{4}=dy^{4}+n_{i}dx^{i},
\mathbf{e}^{5} =dy^{5}+\ ^{1}w_{i}dx^{i}+\ ^{1}w_{3}dv+\ ^{1}w_{4}dy^{4},
\notag \\
&&\mathbf{e}^{6} =dy^{6}+\ ^{1}n_{i}dx^{i}+\ ^{1}n_{3}dv+\ ^{1}n_{4}dy^{4},
\notag \\
&&\mathbf{e}^{7}=dy^{7}+\ ^{2}w_{i}dx^{i}+\ ^{2}w_{3}dv+\ ^{2}w_{4}dy^{4}+\
^{2}w_{5}dy^{5}+\ ^{2}w_{6}dy^{6},  \notag \\
&&\mathbf{e}^{8} =dy^{8}+\ ^{2}n_{i}dx^{i}+\ ^{2}n_{3}dv+\ ^{2}n_{4}dy^{4}+\
^{2}n_{5}dy^{5}+\ ^{2}n_{6}dy^{6},  \notag
\end{eqnarray}%
for nontrivial N--connection coefficients
\begin{eqnarray}
N_{i}^{3} &=&w_{i}(x^{k},v),N_{i}^{4}=n_{i}(x^{k},v);  \label{ncon8d} \\
N_{i}^{5} &=&\ ^{1}w_{i}(x^{k},v,y^{5}),N_{3}^{5}=\
^{1}w_{3}(x^{k},v,y^{5}),N_{4}^{5}=\ ^{1}w_{4}(x^{k},v,y^{5});  \notag \\
N_{i}^{6} &=&\ ^{1}n_{i}(x^{k},v,y^{5});N_{3}^{6}=\
^{1}n_{3}(x^{k},v,y^{5}),N_{4}^{6}=\ ^{1}n_{4}(x^{k},v,y^{5});  \notag \\
N_{i}^{7} &=&\ ^{2}w_{i}(x^{k},v,y^{7}),N_{3}^{7}=\
^{2}w_{3}(x^{k},v,y^{7}),N_{4}^{7}=\ ^{2}w_{4}(x^{k},v,y^{7}),  \notag \\
&&N_{5}^{7}=\ ^{2}w_{3}(x^{k},v,y^{7}),N_{6}^{7}=\ ^{2}w_{4}(x^{k},v,y^{7});
\notag \\
N_{i}^{8} &=&\ ^{2}n_{i}(x^{k},v,y^{7}),N_{3}^{8}=\
^{2}n_{3}(x^{k},v,y^{7}),N_{4}^{8}=\ ^{2}n_{4}(x^{k},v,y^{7}),  \notag \\
&&N_{5}^{8}=\ ^{2}n_{3}(x^{k},v,y^{7}),N_{6}^{8}=\ ^{2}n_{4}(x^{k},v,y^{7}).
\notag
\end{eqnarray}%
} The local coordinates in the above ansatz (\ref{ans8d}) are labeled in the
form $x^{i}=(x^{1},x^{2}),$ for $i,j,...=1,2;$ $y^{3}=v.$

Our goal is to construct and analyze physical implications of solutions of
equations (\ref{einsteqlc}) and (\ref{ensteqcdc}) defined by ansatz (\ref%
{ans8d}) with, respectively, trivial and non--trivial N--connection
coefficients (\ref{ncon8d}).

\subsection{Holonomic brane configurations}

A trapping scenario with diagonal metric from QG with LV to GR can be
constructed for an ansatz of type (\ref{ans8d}) with zero N--connection
coefficients (\ref{ncon8d}) when $h_{5},h_{7},h_{8}=const$ and data $\left[%
g_{i},h_{a}\right] $ define a trivial solution in GR and the local signature
for metrics of of type $(+,-,-,...-)$. Such metrics are written
\begin{eqnarray}
\mathbf{g} &=&\ \phi ^{2}(y^{5})\eta _{\alpha \beta }du^{\alpha }\otimes
du^{\beta }-  \label{ansdiag8d} \\
&&\left( \mathit{l}_{P}\right) ^{2}\overline{h}(y^{5})[\ dy^{5}\otimes \
dy^{5}+dy^{6}\otimes \ dy^{6}\pm dy^{7}\otimes \ dy^{7}\pm dy^{8}\otimes \
dy^{8}],  \notag
\end{eqnarray}%
where $\eta _{\alpha \beta }=diag[1,-1,-1-,1]$ and $\alpha ,\beta
,...=1,2,3,4.$ We shall use also generalized indices of type $\ ^{1}\alpha
=(\alpha ,5,6)$ and $\ ^{2}\alpha =(\ ^{1}\alpha ,7,8),$ respectively for
6--d and 8--d models. Indices of type $\ ^{2}\alpha ,\ ^{2}\beta ,...$ will
run values $1,2,3,4,5,...,m$, where $m\geq 2.$

We consider sources for Einstein equations (\ref{einsteqlc}) with nonzero
components defined by cosmological constant $\Lambda $ and stress--energy
tensor
\begin{equation}
\ _{\shortmid }\Upsilon _{\ \delta }^{\beta }=\Lambda -M^{-(m+2)}\overline{K}%
_{1}(y^{5}),\ _{\shortmid }\Upsilon _{\ 5}^{5}=\ _{\shortmid }\Upsilon _{\
6}^{6}=\Lambda -M^{-(m+2)}\overline{K}_{2}(y^{5}),  \label{source3}
\end{equation}%
for a fundamental mass scale $M$ on $TV,$ $\dim TV=8.$ The fiber coordinates
$y^{5},y^{6},y^{7},y^{8}$ are velocity/momentum type. Diagonal trivial
Finsler brane solutions can be constructed following the methods elaborated
(for extra dimensional gravity) in Refs. \cite{midod,gm2,gs2,singlbr}.%
\footnote{%
In this paper we shall use some adapted classes of solutions from the just
cited paper where the extra dimensions (2,3 etc) are analyzed in general
form. Here we also note that our notations for Finsler gravity models on
tangent bundles are different for those used in the above papers on 6-d, and
other dimensions, brane gravity solutions.} A metric (\ref{ansdiag8d}) is a
solution of (\ref{einsteqlc}) if
\begin{equation}
\phi ^{2}(y^{5})=\frac{3\epsilon ^{2}+a(y^{5})^{2}}{3\epsilon
^{2}+(y^{5})^{2}}\mbox{ and }\mathit{l}_{P}\sqrt{|\overline{h}(y^{5})|}=%
\frac{9\epsilon ^{4}}{\left[ 3\epsilon ^{2}+(y^{5})^{2}\right] ^{2}},
\label{cond1}
\end{equation}%
where $a$ is an integration constant and the width of brane is $\epsilon ,$
with some fixed integration parameters when $\frac{\partial ^{2}\phi }{%
\partial (y^{5})^{2}}\mid _{y^{5}=\epsilon }=0$ and $\mathit{l}_{P}\sqrt{|%
\overline{h}(y^{5})|}\mid _{y^{5}=0}=1;$ this states the conditions that on
diagonal branes the Minkowski metric on $TV$ is 6--d or 8--d.  We get
compatible (with field equations) sources (\ref{source3}) if {\small
\begin{eqnarray}
&&\overline{K}_{1}(y^{5})M^{-(m+2)}=\Lambda +\left[ 3\epsilon
^{2}+(y^{5})^{2}\right] ^{-2}[\frac{2am(a(m+2)-3)}{3\epsilon ^{2}}%
(y^{5})^{4}+  \notag \\
&&2[-2a(m^{2}+2m+6)+3(m+3)(1+a^{2})](y^{5})^{2} -6\epsilon ^{2}m(m-3a+2)],
\label{cond2} \\
&&\overline{K}_{2}(y^{5})M^{-(m+2)} =\Lambda +\left[ 3\epsilon
^{2}+(y^{5})^{2}\right] ^{-2}[\frac{2a(m-1)(a(m+2)-4)}{3\epsilon ^{2}}%
(y^{5})^{4} +  \notag \\
&&4[-a(m^{2}+m+10)+2(m+2)(1+a^{2})](y^{5})^{2} -6\epsilon ^{2}(m-1)(m-4a+2)].
\notag
\end{eqnarray}
} The above formulas for $m=2$ are similar to those for usual 6--d diagonal
brane solutions with that difference that in our case the width $\epsilon
^{2}=40M^{4}/3\Lambda $ is with extra velocity/momentum coordinates and
certain constants are related to $\mathit{l}_{P}.$ Here we also emphasize
that $y^{5} $ has a finite maximal value $y^{5}_0 $ on $TM$ because the
light velocity is finite.

The Einstein equations (\ref{einsteqlc}) are for the Levi--Civita connection
when {\small
\begin{equation}
\nabla _{\ ^{2}\alpha }\ _{\shortmid }\Upsilon _{\ \quad }^{\ ^{2}\alpha \
^{2}\beta }=(\sqrt{|\ ^{F}\mathbf{g}|})^{-1}\mathbf{e}_{\ ^{2}\alpha }(\sqrt{%
|\ ^{F}\mathbf{g}|}\ _{\shortmid }\Upsilon _{\ \quad }^{\ ^{2}\alpha \
^{2}\beta })+\ _{\shortmid }\Gamma _{\ \ ^{2}\alpha \ ^{2}\gamma }^{\
^{2}\beta }\ _{\shortmid }\Upsilon _{\ \quad }^{\ ^{2}\alpha \ ^{2}\gamma
}=0.  \label{cond3}
\end{equation}%
} For our ansatz (\ref{ansdiag8d}) and (\ref{source3}) with coefficients (%
\ref{cond1}) and (\ref{cond2}), such a conservation law is satisfied if
\begin{equation}
\frac{\partial \overline{K}_{1}}{\partial (y^{5})}=4\left( \overline{K}_{2}-%
\overline{K}_{1}\right) \frac{\partial \ln |\phi |}{\partial (y^{5})}.
\label{cond3a}
\end{equation}

We conclude that a metric (\ref{ansdiag8d}), when the coefficients are
subjected to conditions (\ref{cond1}) -- (\ref{cond3a}), defines trapping
solutions containing "diagonal" extensions of GR to a 8--d $TM$ and/or
possible restrictions to 6--d and 7--d. Such solutions provide also
mechanisms of corresponding gravitational trapping for fields of spins $%
0,1/2,1,2$ (see similar proofs in Refs. \cite{midod,gm2,gs2,singlbr}). The
above results are in some sense expected since for diagonal configurations
our model is similar to the 6--d and higher dimension ones constructed in
the mentioned papers. There are two substantial differences that $m$ is
fixed to have a maximal value $m=4$ and that $y^{5}\leq $ $y_{0}^{5}$ where $%
y_{0}^{5}$ is determined by the maximal speed of light propagation (in the
supposition that it is the same as for propagation of gravitational
interactions in QG).

The behavior of physical suitable sources determined by ansatz $\overline{K}%
_{1}(y^{5})$ and/or $\overline{K}_{2}(y^{5})$ depends (for this class of
solutions) on four parameters, $m,\epsilon ,\Lambda $ and $a.$ This is quite
surprising for QG and solutions with LV because usually it is expected that
quantum effects and/or Lorenz violations may be important for distances $%
\sim \mathit{l}_{P}. $ It is possible to have either $\overline{K}%
_{1}(y^{5}) $ or $\overline{K}_{2}(y^{5}),$ or both, go to zero for
corresponding choices of the mentioned four parameters. To see this we may
use the analysis from the Conclusion section of Ref. \cite{singlbr} even on $%
TM$ with finite $y_{0}^{5}$ it is not necessary to consider $%
y^{5}\rightarrow \infty .$ For $m>4,$ which is not the case of Finsler
geometry from QG dispersions, the function $\phi ^{2}(y^{5})$ may become
singular at $y^{5}\rightarrow \infty .$ Such problems can be avoided because
for Finsler configurations derived from GR we can take always $m=1,2,3$
and/or consider $y_{0}^{5}.$ We can consider that in Finsler gravity that $%
\overline{K}_{1}(y^{5})\rightarrow 0$ or $\overline{K}_{2}(y^{5})\rightarrow
0$ for $y^{5}\rightarrow y_{0}^{5}.$

We do not address the question of stability of Finsler brane solutions in
this work. In general, stabile configurations can be constructed for
diagonal solutions which survive for nonholonomically constrained
off--diagonal ones (proofs are similar to those for extra dimensional brane
solutions; we shall address the problem in details in our further works).

\subsection{Finsler brane solutions}

One of the main purposes of this work is to elaborate trapping scenarios for
Finsler configurations with positively nontrivial N--connections as
solutions of nonholonomic gravitational equations (\ref{ensteqcdc}). The
priority of such generic off--diagonal solutions is that they allow us to
distinguish the QG phenomenology and effects with LV of (pseudo) Finsler
type from that described, on $TV$ by (pseudo) Riemannian ones.

\subsubsection{Decoupling of equations in Einstein--Finsler gravity}

We consider an ansatz (\ref{ans8d}) multiplied to $\phi ^{2}(y^{5})$ and
with non--trivial N--connection coefficients (\ref{ncon8d}) and define the
conditions when the coefficients generate exact solutions of (\ref{ensteqcdc}%
) we get extending the solutions and sources (\ref{source3}). The sources
are parametrized in a form similar to (\ref{sourceans}),
\begin{eqnarray}
\widetilde{\mathbf{\Upsilon }}_{\ \delta }^{\beta } &=&diag[\widetilde{%
\mathbf{\Upsilon }}_{\ 1}^{1}=\widetilde{\mathbf{\Upsilon }}_{\ 2}^{2}=%
\widetilde{\mathbf{\Upsilon }}_{2}(u^{\ ^{2}\alpha }),\widetilde{\mathbf{%
\Upsilon }}_{\ 3}^{3}=\widetilde{\mathbf{\Upsilon }}_{\ 4}^{4}=\widetilde{%
\mathbf{\Upsilon }}_{4}(u^{\ ^{2}\alpha }),  \notag \\
&&\widetilde{\mathbf{\Upsilon }}_{\ 5}^{5}=\widetilde{\mathbf{\Upsilon }}_{\
6}^{6}=\widetilde{\mathbf{\Upsilon }}_{6}(u^{\ ^{2}\alpha }),\widetilde{%
\mathbf{\Upsilon }}_{\ 7}^{7}=\widetilde{\mathbf{\Upsilon }}_{\ 8}^{8}=%
\widetilde{\mathbf{\Upsilon }}_{8}(u^{\ ^{2}\alpha })],  \label{sourcb}
\end{eqnarray}%
when the coefficients are subjected to algebraic conditions (for vanishing
N---coefficients, containing respectively the functions (\ref{source3})
determining sources in the gravitational field equations) $\ ^{h}\Lambda
(x^{i}) = \widetilde{\mathbf{\Upsilon }}_{4}+\widetilde{\mathbf{\Upsilon }}%
_{6}+\widetilde{\mathbf{\Upsilon }}_{8},\ ^{v}\Lambda (x^{i},v)=\widetilde{%
\mathbf{\Upsilon }}_{2}+\widetilde{\mathbf{\Upsilon }}_{6}+\widetilde{%
\mathbf{\Upsilon }}_{8}, \ ^{5}\Lambda (x^{i},y^{5}) = \widetilde{\mathbf{%
\Upsilon }}_{2}+\widetilde{\mathbf{\Upsilon }}_{4}+\widetilde{\mathbf{%
\Upsilon }}_{8},\ \ ^{7}\Lambda (x^{i},y^{5},y^{7})=\widetilde{\mathbf{%
\Upsilon }}_{2}+\widetilde{\mathbf{\Upsilon }}_{4}+\widetilde{\mathbf{%
\Upsilon }}_{6}$.

Using the above assumptions on metric ansatz and sources, the conditions of
Theorem \ref{th3} can be extended step by step for dimensions 2+2+2+2. We
obtain a system of equations with decoupling (separation) of partial
differential equations (generalizing respectively (\ref{4ep1a}) and (\ref%
{4ep4a})): {\small
\begin{eqnarray}
\widetilde{R}_{1}^{1} &=&\widetilde{R}_{2}^{2} =\frac{1}{2g_{1}g_{2}}[\frac{%
g_{1}^{\bullet }g_{2}^{\bullet }}{2g_{1}}+\frac{(g_{2}^{\bullet })^{2}}{%
2g_{2}}-g_{2}^{\bullet \bullet }+\frac{g_{1}^{^{\prime }}g_{2}^{^{\prime }}}{%
2g_{2}}+\frac{(g_{1}^{^{\prime }})^{2}}{2g_{1}}-g_{1}^{^{\prime \prime
}}]=-\ ^{h}\Lambda (x^{i}),  \label{4ep1b} \\
\widetilde{R}_{3}^{3} &=&\widetilde{R}_{4}^{4}=\frac{1}{2h_{3}h_{4}}%
[-h_{4}^{\ast \ast }+\frac{\left( h_{4}^{\ast }\right) ^{2}}{2h_{4}}+\frac{%
h_{3}^{\ast }\ h_{4}^{\ast }}{2h_{3}}]=\ -\ ^{v}\Lambda (x^{i},v),
\label{4ep2b} \\
\widetilde{R}_{5}^{5} &=&\widetilde{R}_{6}^{6}=\frac{1}{2h_{5}h_{6}}%
[-\partial _{y^{5}y^{5}}^{2}h_{6}+\frac{\left( \partial _{y^{5}}h_{6}\right)
^{2}}{2h_{6}}+\frac{(\partial _{y^{5}}h_{5})\ (\partial _{y^{5}}h_{6})}{%
2h_{5}}] =-\ ^{5}\Lambda (x^{i},y^{5}),  \notag \\
\widetilde{R}_{7}^{7} &=&\widetilde{R}_{8}^{8}=\frac{1}{2h_{7}h_{8}}%
[-\partial _{y^{7}y^{7}}^{2}h_{8}+\frac{\left( \partial _{y^{7}}h_{8}\right)
^{2}}{2h_{8}}+\frac{(\partial _{y^{7}}h_{7})\ (\partial _{y^{7}}h_{8})}{%
2h_{7}}] = -\ ^{7}\Lambda (x^{i},y^{5},y^{7}),  \notag
\end{eqnarray}%
} with partial derivatives on velocity/momentum type coordinates taken on
respective fibers, for instance, $\partial _{y^{5}}h_{6}=\partial
h_{6}/\partial y^{5}.$ The equations (\ref{4ep1b}) are completely similar to
(\ref{4ep1a}) and the equations (\ref{4ep2b}) reproduce three times
(correspondingly, for couples of variables $\left( y^{3}=v,y^{4}\right)
,\left( y^{5},y^{6}\right) ,\left( y^{7},y^{8}\right) ,$ and ''anisotropic''
coordinates $v,y^{5},y^{7\text{ }}$and Killing symmetries on vectors $%
\partial /\partial y^{4},$ $\partial /\partial y^{6}$ and $\partial
/\partial y^{8})$ the equations (\ref{4ep2a}). The equations {\small
\begin{eqnarray}
\ \widetilde{R}_{3j} &=&\frac{h_{3}^{\ast }}{2h_{3}}w_{j}^{\ast }+A^{\ast
}w_{j}+B_{j}=0,  \notag \\
\ \widetilde{R}_{5j} &=& \frac{\partial _{y^{5}}h_{5}}{2h_{5}}\partial
_{y^{5}}\ ^{1}w_{j}+\left( \partial _{y^{5}}\ ^{1}A\right) \ ^{1}w_{j}+\
^{1}B_{j}=0,  \notag \\
\widetilde{R}_{7j} &=&\frac{\partial _{y^{7}}h_{7}}{2h_{7}}\partial
_{y^{7}}\ ^{2}w_{j}+\left( \partial _{y^{7}}\ ^{2}A\right) \ ^{2}w_{j}+\
^{2}B_{j}=0,  \label{4ep3b}
\end{eqnarray}%
} generalize on 8--d $TM$ the equations (\ref{4ep3a}). The system {\small
\begin{eqnarray}
\widetilde{R}_{4i} &=&-\frac{h_{4}^{\ast }}{2h_{3}}n_{i}^{\ast }+\frac{%
h_{4}^{\ast }}{2}K_{i}=0,\ \widetilde{R}_{6i} = -\frac{\partial _{y^{5}}h_{6}%
}{2h_{5}}\partial _{y^{5}}\ ^{1}n_{i}+\frac{\partial _{y^{5}}h_{6}}{2}\
^{1}K_{i}=0,  \notag \\
\widetilde{R}_{8i} &=&-\frac{\partial _{y^{7}}h_{8}}{2h_{7}}\partial
_{y^{7}}\ ^{2}n_{i}+\frac{\partial _{y^{7}}h_{8}}{2}\ ^{2}K_{i}=0,
\label{4ep4b}
\end{eqnarray}%
} is an extension of (\ref{4ep4a}).

In the above formulas (\ref{4ep3b}) and (\ref{4ep4b}), there are considered
nontrivial N--connection coefficients (\ref{ncon8d}) and  and extensions of
of (\ref{aux}), {\small
\begin{eqnarray*}
&&A =(\frac{h_{3}^{\ast }}{2h_{3}}+\frac{h_{4}^{\ast }}{2h_{4}}) ,\ B_{k}=%
\frac{h_{4}^{\ast }}{2h_{4}}( \frac{\partial _{k}g_{1}}{2g_{1}}-\frac{%
\partial _{k}g_{2}}{2g_{2}}) -\partial _{k}A, K_{1} =-\frac{1}{2}(\frac{%
g_{1}^{\prime }}{g_{2}h_{3}}+\frac{g_{2}^{\bullet }}{g_{2}h_{4}}) , \\
&& K_{2}=\frac{1}{2}(\frac{g_{2}^{\bullet }}{g_{1}h_{3}}-\frac{g_{2}^{\prime
}}{g_{2}h_{4}}); \ ^{1}A =(\frac{\partial _{y^{5}}h_{5}}{2h_{5}}+\frac{%
\partial _{y^{5}}h_{6}}{2h_{6}}) , \ ^{1}B_{k}=\frac{\partial _{y^{5}}h_{6}}{%
2h_{6}}(\frac{\partial _{k}g_{1}}{2g_{1}}-\frac{\partial _{k}g_{2}}{2g_{2}})
\\
&&-\partial _{k}\ ^{1}A, \ ^{1}K_{1}= -\frac{1}{2}(\frac{g_{1}^{\prime }}{%
g_{2}h_{5}}+\frac{g_{2}^{\bullet }}{g_{2}h_{6}}) ,\ ^{1}K_{2}=\frac{1}{2}(%
\frac{g_{2}^{\bullet }}{g_{1}h_{5}}-\frac{g_{2}^{\prime }}{g_{2}h_{6}}); \\
&& \ ^{2}A =(\frac{\partial _{y^{7}}h_{7}}{2h_{7}}+\frac{\partial
_{y^{7}}h_{8}}{2h_{8}}),\ ^{2}B_{k}=\frac{\partial _{y^{7}}h_{8}}{2h_{8}}(%
\frac{\partial _{k}g_{1}}{2g_{1}}-\frac{\partial _{k}g_{2}}{2g_{2}})
-\partial _{k}\ ^{2}A, \\
&&\ ^{2}K_{1}=-\frac{1}{2}(\frac{g_{1}^{\prime }}{g_{2}h_{7}}+\frac{%
g_{2}^{\bullet }}{g_{2}h_{8}}),\ ^{2}K_{2}=\frac{1}{2}(\frac{g_{2}^{\bullet }%
}{g_{1}h_{7}}-\frac{g_{2}^{\prime }}{g_{2}h_{8}}).
\end{eqnarray*}
}

\subsubsection{Integration of equations}

The conditions of Theorem \ref{th4} can be extended on 8--d $TM$ which
allows us to integrate in general forms the system of gravitational field
equations (see respectively the equations (\ref{4ep1b}), (\ref{4ep2b}), (\ref%
{4ep3b}) and (\ref{4ep4b}) in EFG. Such solutions can be parametrized
additionally to the data (\ref{sol1})--(\ref{sol4a}) (for $%
g_{i}(x^{k}),h_{a}(x^{k},v),w_{i}(x^{k},v)$ and $n_{i}(x^{k},v))$ by
coefficients {\small
\begin{eqnarray}
h_{5}(x^{i},y^{5}) &=&\epsilon _{5}\ _{1}^{0}h(x^{i})\ [\partial _{y^{5}}\
^{1}f(x^{i},y^{5})]^{2}|\ ^{1}\varsigma (x^{i},y^{5})|,\ ^{1}\varsigma =\
_{1}^{0}\varsigma (x^{i})-\frac{\epsilon _{5}}{8}\ _{1}^{0}  \notag \\
h(x^{i}) &&\int (dy^{5})\ ^{5}\Lambda (x^{i},y^{5})\ [\partial _{y^{5}}\
^{1}f(x^{i},y^{5})]\ [\ ^{1}f(x^{i},y^{5})-\ _{1}^{0}f(x^{i})],  \notag \\
h_{6}(x^{i},y^{5}) &=&\epsilon _{6}[\ ^{1}f(x^{i},y^{5})-\
_{1}^{0}f(x^{i})]^{2};  \notag
\end{eqnarray}
\begin{eqnarray}
\ ^{1}w_{j}(x^{i},y^{5}) &=&\ _{0}^{1}w_{j}(x^{i})\exp \left\{
-\int_{0}^{y^{5}}\left[ \frac{2h_{5}\partial _{y^{5}}(\ ^{1}A)}{\partial
_{y^{5}}h_{5}}\right] _{y^{5}\rightarrow v_{1}}dv_{1}\right\}  \label{data2a}
\\
\int_{0}^{y^{5}}dv_{1} &&\left[ \frac{h_{5}\ ^{1}B_{j}}{\partial
_{y^{5}}h_{5}}\right] _{y^{5}\rightarrow v_{1}}\exp \left\{
-\int_{_{0}}^{v_{1}}\left[ \frac{2h_{5}\partial _{y^{5}}\ ^{1}A}{\partial
_{y^{5}}h_{5}}\right] _{y^{5}\rightarrow v_{1}}dv_{1}\right\} ,  \notag \\
\ ^{1}n_{j}(x^{i},y^{5}) &=&\ _{0}^{1}n_{j}(x^{k})+\int dy^{5}\ h_{5}\
^{1}K_{j},  \notag \\
h_{7}(x^{i},y^{5},y^{7}) &=&\epsilon _{7}\ _{2}^{0}h(x^{i})\ [\partial
_{y^{7}}\ ^{2}f(x^{i},y^{5},y^{7})]^{2}|\ ^{2}\varsigma
(x^{i},y^{5},y^{7})|,\ ^{2}\varsigma =\ _{2}^{0}\varsigma (x^{i})  \notag \\
-\frac{\epsilon _{7}}{8}\ _{2}^{0}h(x^{i}) &&\int (dy^{7})\ ^{7}\Lambda
\lbrack \partial _{y^{7}}\ ^{2}f(x^{i},y^{5},y^{7})]\ [\
^{2}f(x^{i},y^{5},y^{7})-\ _{2}^{0}f(x^{i})],  \notag \\
h_{8}(x^{i},y^{5},y^{7}) &=&\epsilon _{8}[\ ^{2}f(x^{i},y^{5},y^{7})-\
_{2}^{0}f(x^{i})]^{2};  \notag \\
\ ^{2}w_{j}(x^{i},y^{5},y^{7}) &=&\ _{0}^{2}w_{j}(x^{i})\exp \left\{
-\int_{0}^{y^{7}}\left[ \frac{2h_{7}\partial _{y^{7}}(\ ^{2}A)}{\partial
_{y^{7}}h_{7}}\right] _{v\rightarrow v_{1}}dv_{1}\right\}
\int_{0}^{y^{7}}dv_{1}  \notag \\
&&\left[ \frac{h_{7}\ ^{2}B_{j}}{\partial _{y^{7}}h_{7}}\right]
_{y^{7}\rightarrow v_{1}}\exp \left\{ -\int_{_{0}}^{v_{1}}\left[ \frac{%
2h_{7}\partial _{y^{7}}\ ^{2}A}{\partial _{y^{7}}h_{7}}\right]
_{y^{7}\rightarrow v_{1}}dv_{1}\right\} ,  \notag \\
\ ^{2}n_{j}(x^{i},y^{7}) &=&\ _{0}^{2}n_{j}(x^{k})+\int dy^{7}\ h_{7}\
^{2}K_{j}  \notag
\end{eqnarray}%
} Such solutions with nonzero $h_{3}^{\ast },$ $h_{4}^{\ast },\partial
_{y^{5}}h_{5},\partial _{y^{5}}h_{6},\partial _{y^{7}}h_{7},\partial
_{y^{7}}h_{8}$ are determined by generating functions $f(x^{i},v),f^{\ast
}\neq 0,$ $\ ^{1}f(x^{i},y^{5}),\partial _{y^{5}}\ ^{1}f\neq 0,$ \newline
$\ ^{2}f(x^{i},y^{5},y^{7}),$ $\partial _{y^{7}}\ ^{2}f\neq 0,$ and
integration functions $\ ^{0}f(x^{i}),\ ^{0}h(x^{i}),$ $\ _{0}w_{j}(x^{i}),$
$\ _{0}n_{i}(x^{k}),\ _{1}^{0}f(x^{i}),\ _{1}^{0}h(x^{i}),\
_{0}^{1}w_{j}(x^{i}),\ _{0}^{1}n_{i}(x^{k}),\ _{2}^{0}f(x^{i}),\
_{2}^{0}h(x^{i}),\ _{0}^{2}w_{j}(x^{i}),\ _{0}^{2}n_{i}(x^{k}).$ We should
chose and/or fix such functions following additional assumptions on symmetry
of solutions, boundary conditions etc.

There are substantial differences between branes in Finsler gravity and in
extra dimension theories. In the first case, the physical constants/
parameters are induced in quasi--classical limits from QG on (co) tangent
bundles but in the second case the constructions are for high dimensional
spacetime models. An important problem to be solved for such geometries is
to show that there are trapping mechanisms for nonholonomic configurations
to Finsler branes with finite widths (determined by the maximal value of
light velocity) and possible warping on ''fiber'' coordinates.

\subsubsection{On (non) diagonal brane solutions on $TM$}

It is not clear what physical interpretation may have the above general
solutions for Finsler gravity. We have to impose additional restrictions on
some coefficients of metrics and sources in order to construct in explicit
form certain Finsler brane configurations and model a trapping mechanism
with generic off--diagonal metrics.

Let us consider a class of sources in EFG when for trivial N--connection
coefficients (i.e. for zero values (\ref{ncon8d})) the sources $\widetilde{%
\mathbf{\Upsilon }}_{\ \ ^{2}\delta }^{\ ^{2}\beta }$ (\ref{sourcb})
transform into data $\ _{\shortmid }\Upsilon _{\ \ ^{2}\delta }^{\ ^{2}\beta
}$ (\ref{source3}), with nontrivial limits for $\ _{\shortmid }\Upsilon _{\
\delta }^{\beta }=\Lambda -M^{-(m+2)}\overline{K}_{1}(y^{5})$ and $\
_{\shortmid }\Upsilon _{\ 5}^{5}=\ _{\shortmid }\Upsilon _{\ 6}^{6}=\Lambda
-M^{-(m+2)}\overline{K}_{2}(y^{5})$. The generating $f$--functions are taken
in the form when $h_{5} =\mathit{l}_{P}\frac{\overline{h}(y^{5})}{\phi
^{2}(y^{5})}\ ^{q}h_{5}(x^{i},y^{5}),\ h_{6} =\mathit{l}_{P}\frac{\overline{h%
}(y^{5})}{\phi ^{2}(y^{5})}\ ^{q}h_{6}(x^{i},y^{5}), h_{7} = \mathit{l}_{P}%
\frac{\overline{h}(y^{5})}{\phi ^{2}(y^{5})}\ ^{q}h_{7}(x^{i},y^{5},y^{7}),$
$h_{8} =\mathit{l}_{P}\frac{\overline{h}(y^{5})}{\phi ^{2}(y^{5})}\
^{q}h_{8}(x^{i},y^{5},y^{7})$, where the generating functions are
parametrized in such a form that $\phi ^{2}(y^{5})$ and $h_{5}(y^{5})$ are
those for diagonal metrics, i.e. of type (\ref{cond1}), and $\ ^{q}h_{5},\
^{q}h_{6},\ ^{q}h_{7},$ $\ ^{q}h_{8}$ are computed following formulas (\ref%
{sol1})--(\ref{sol4a}) and (\ref{data2a}). The resulting off--diagonal
solutions are {\small
\begin{eqnarray}
\mathbf{g} &=&g_{1}dx^{1}\otimes dx^{1}+g_{2}dx^{2}\otimes dx^{2}+h_{3}%
\mathbf{e}^{3}{\otimes }\mathbf{e}^{3}\ +h_{4}\mathbf{e}^{4}{\otimes }%
\mathbf{e}^{4}\ +  \label{fbr} \\
&& \left( \mathit{l}_{P}\right) ^{2}\frac{\overline{h}}{\phi ^{2}} [\
^{q}h_{5}\mathbf{e}^{5}\otimes \ \mathbf{e}^{5}+\ ^{q}h_{6}\mathbf{e}%
^{6}\otimes \ \mathbf{e}^{6}+\ ^{q}h_{7}\mathbf{e}^{7}\otimes \ \mathbf{e}%
^{7}+\ ^{q}h_{8}\mathbf{e}^{8}\otimes \ \mathbf{e}^{8}],  \notag \\
\mathbf{e}^{3} &=&dy^{3}+w_{i}dx^{i},\mathbf{e}^{4}=dy^{4}+n_{i}dx^{i}, \
\mathbf{e}^{5} = dy^{5}+\ ^{1}w_{i}dx^{i},  \label{ncfbr} \\
\mathbf{e}^{6} &=& dy^{6}+\ ^{1}n_{i}dx^{i},\ \mathbf{e}^{7} = dy^{7}+\
^{2}w_{i}dx^{i},\mathbf{e}^{8}=dy^{8}+\ ^{2}n_{i}dx^{i}.  \notag
\end{eqnarray}
} Any solution of type (\ref{fbr}) describes an off--diagonal trapping for
8--d (respectively, for corresponding classes of generating and integration
functions, 5--, 6--, 7--d) to 4--d modifications of GR with some corrections
depending on QG ''fluctuations'' and LV effects. There is a class of sources
when for vanishing N--connection coefficients (\ref{ncfbr}) we get diagonal
metrics of type (\ref{ans8d}) but multiplied to a conformal factor $\phi
^{2}(y^{5})$ \ when the $h$--coefficients are solutions of equations of type
(\ref{4ep2b}). Even for some diagonal limits, such metrics metrics are very
different and can not be transformed, in general form, from one to another
even asymptotically, when $\phi (y^{5})\rightarrow a$ for $y^{5}\rightarrow
\infty ,$ they may mimic some similar behavior and QG contributions.

With respect to a local coordinate cobase $du^{\ ^{2}\alpha
}=(dx^{i},dy^{a},dy^{\ ^{1}a},dy^{\ ^{2}a}),$ a solution (\ref{fbr}) is
parametrized by an off--diagonal matrix $g_{\ ^{2}\alpha \ ^{2}\beta }=$%
{\small
\begin{equation*}
\left[
\begin{array}{cccccccc}
A_{11} & A_{12} & w_{1}h_{3} & n_{1}h_{4}+ & \ ^{1}w_{1}h_{5} & \
^{1}n_{1}h_{6}+ & \ ^{2}w_{1}h_{7} & \ ^{2}n_{1}h_{8} \\
A_{21} & A_{22} & w_{2}h_{3} & n_{2}h_{4} & \ ^{1}w_{2}h_{5} & \
^{1}n_{2}h_{6} & \ ^{2}w_{2}h_{7} & \ ^{2}n_{2}h_{8} \\
w_{1}h_{3} & w_{2}h_{3} & h_{3} & 0 & 0 & 0 & 0 & 0 \\
n_{1}h_{4} & n_{2}h_{4} & 0 & h_{4} & 0 & 0 & 0 & 0 \\
\ ^{1}w_{1}h_{5} & ^{1}w_{2}h_{5} & 0 & 0 & h_{5} & 0 & 0 & 0 \\
\ ^{1}n_{1}h_{6} & \ ^{1}n_{2}h_{6} & 0 & 0 & 0 & h_{6} & 0 & 0 \\
\ ^{2}w_{1}h_{7} & \ ^{2}w_{2}h_{7} & 0 & 0 & 0 & 0 & h_{7} & 0 \\
\ ^{2}n_{1}h_{8} & \ ^{2}n_{2}h_{8} & 0 & 0 & 0 & 0 & 0 & h_{8}%
\end{array}%
\right]
\end{equation*}%
} where the possible observable QG and LV contributions (fluctuations in
general form) are distinguished by terms proportional to $\left( \mathit{l}%
_{P}\right) ^{2}$ in
\begin{eqnarray*}
A_{11} &=&g_{1}+w_{1}^{2}h_{3}+n_{1}^{2}h_{4}+\left( \mathit{l}_{P}\right)
^{2}\frac{\overline{h}}{\phi ^{2}}\times \\
&&\left[ (\ ^{1}w_{1})^{2}\ ^{q}h_{5}+(\ ^{1}n_{1})^{2}\ ^{q}h_{6}+(\
^{2}w_{1})^{2}\ ^{q}h_{7}+(\ ^{2}n_{1})^{2}\ ^{q}h_{8}\right] , \\
A_{12} &=&A_{21}=w_{1}w_{2}h_{3}+n_{1}n_{2}h_{4}+\left( \mathit{l}%
_{P}\right) ^{2}\frac{\overline{h}}{\phi ^{2}}\times \\
&&\left[ \ ^{1}w_{1}\ ^{1}w_{2}\ ^{q}h_{5}+\ ^{1}n_{1}\ ^{1}n_{2}\
^{q}h_{6}+\ ^{2}w_{1}\ ^{2}w_{2}\ ^{q}h_{7}+\ ^{2}n_{1}\ ^{2}n_{2}\ ^{q}h_{8}%
\right] , \\
A_{22} &=&g_{2}+w_{2}^{2}h_{3}+n_{2}^{2}h_{4}+\left( \mathit{l}_{P}\right)
^{2}\frac{\overline{h}}{\phi ^{2}}\times \\
&&\left[ (\ ^{1}w_{2})^{2}\ ^{q}h_{5}+(\ ^{1}n_{2})^{2}\ ^{q}h_{6}+(\
^{2}w_{2})^{2}\ ^{q}h_{7}+(\ ^{2}n_{2})^{2}\ ^{q}h_{8}\right] .
\end{eqnarray*}%
It is possible to distinguish \ experimentally such off--diagonal metrics in
Finsler geometry from diagonal configurations (\ref{ansdiag8d}) with
Levi--Civita connection on $TM.$

On Finsler branes determined by data (\ref{data2a}), the gravitons are
allowed to propagate in the bulk of a Finsler spacetime with dependence on
velocity/ momentum coordinates. The reason to introduce warped Finsler
geometries and consider various trapping mechanisms is that following modern
experimental data there are not explicit observations for Finsler like
metrics in gravity even such dependencies can be always derived in various
QG models. There are expectations that brane trapping effects may allow us
to detect QG and LV effects experimentally even at scales much large than
the Planck one and for different scenarios than those considered in Refs. %
\cite{kost4,liberati,amelino,dimopoulos,anchor,lammer,stavr0,girelli}.

It is not surprising that two classes of solutions, of type (\ref{fbr}) and (%
\ref{ans8d}), are very different on structure and physical implications
because such metrics were subjected to the conditions to solve two different
classes of gravitational field equations, respectively, (\ref{ensteqcdc})
and (\ref{einsteqlc}). It should also emphasized here that conservation laws
of type (\ref{cond3}) are not satisfied for Finsler type solutions even the
conditions (\ref{cond3a}) can be imposed for some initial data for $%
\overline{K}_{1},\overline{K}_{2}$ and $\phi .$ In EFG with the Cartan
d--connection, the conservation law $\nabla _{\ ^{2}\alpha }\ _{\shortmid
}\Upsilon _{\ \quad }^{\ ^{2}\alpha \ ^{2}\beta }=0$ is nonholonomically
deformed into {\small
\begin{equation}
(\sqrt{|\ ^{F}\mathbf{g}|})^{-1}\mathbf{e}_{\ ^{2}\alpha }(\sqrt{|\ ^{F}%
\mathbf{g}|}\ _{\shortmid }\Upsilon _{\ \quad }^{\ ^{2}\alpha \ ^{2}\beta
})+\ _{\shortmid }\Gamma _{\ \ ^{2}\alpha \ ^{2}\gamma }^{\ ^{2}\beta }\
_{\shortmid }\Upsilon _{\ \quad }^{\ ^{2}\alpha \ ^{2}\gamma }=\widetilde{%
\mathbf{Z}}_{\ \ ^{2}\alpha \ ^{2}\gamma }^{\ ^{2}\beta }\ _{\shortmid
}\Upsilon _{\ \quad }^{\ ^{2}\alpha \ ^{2}\gamma },  \label{const}
\end{equation}%
} where the distortion term $\widetilde{\mathbf{Z}}_{\ \ ^{2}\alpha \
^{2}\gamma }^{\ ^{2}\beta }$ (\ref{dist}) is determined by nontrivial
torsion components (\ref{nztors}) (in their turn completely defined by
generic off--diagonal terms of $\mathbf{g}$ and respective N--connection
coefficients).

Conservation laws of type (\ref{const}) are typical for systems with some
degrees of freedom subjected to anholonomic constraints, see (\ref{anhrel}),
which is the case of Finsler spaces (more than that, possible LV result in
more sophisticate local spacetime symmetries). They are derived from
generalized Bianchi equations for the normal/Cartan d--connection $%
\widetilde{\mathbf{D}} $ \cite{ma} (the problem of formulating conservation
laws for gravity theories with local gravity is discussed in Refs. \cite%
{vrevflg,vcrit,vsgg}). In our approach, with respect to N--adapted frames,
we can compute constraints of type (\ref{cond3a}) with some additional terms
following from (\ref{const}) reflecting some arbitrariness for fixing
nonholonomic distributions and frames on $TM$ for EFG. In general, the
anholonomic deformation method allows us to construct Finsler brane type
solutions with generic off--diagonal QG and LV terms expressed in general
form (not depending explicitly on the type of metric compatible
d--connection we consider, type of fundamental Finsler function and
generating/integration functions).

\section{Discussion and Conclusions}

\label{s4} During the last decade, Finsler like gravity models were studied
because they appeared to provide possible scenarios of  Lorentz symmetry
violations (LV) in quantum gravity (QG),  new  ideas for modified gravity theories with
local anisotropy and in relation to dark matter and dark energy problems in
modern cosmology \cite%
{liberati,mavromatos1,stavr0,girelli,gibbons,stavr1,mavromatos,sindoni,stavr2,lichang4,lichang6}%
. The crux of the argument that QG can be related to  Finsler
geometry  follows from three important physical results:
\begin{enumerate}
\item There are fundamental uncertainty relations for quantum physics,
\begin{equation}
x_{k}p_{j}-p_{k}x_{j}\sim i\hbar ,  \label{uncert}
\end{equation}%
where $x_{i}$ are operators associated to coordinates on a manifold $V$ and
$p_{j}$ are momentum variables associated to $T^{\ast }V,$ being dual to
certain ''velocities'' $y_{k}$ on $TV;i^{2}=-1$ and $\hbar $ is the Planck
constant.
\item The bulk of  QG theories are with nonlinear dispersion relations (\ref{disp}) which encode certain Finsler structure of type (\ref{fbm}).

\item The general relativity (GR) theory can be written equivalently in
so--called "formal" Finsler variables which can be defined on any (pseudo) Riemannian
manifold with conventional horizontal (h) and (v) vertical splitting (for
instance, via non--integrable 2+2 distributions/frame decompositions) \cite%
{vexsol,gensol2,vrevflg,vsgg}.
\end{enumerate}

Quantum theories are, at least in quasi--classical
limits, some geometric models on  (co) tangent bundles of
certain manifolds endowed with geometric, dynamical and
nonholonomic structure adapted to  non--integrable
distributions on $TV,$ or $T^{\ast }V$,  determined by
 generating functions  $F(x,y)$ (in particular,  of uncertainty type (\ref%
{uncert})). The principle of equivalence in GR imposes via nonlinear
dispersion relations  (\ref{fbm}) and (\ref{disp}) the condition that $%
F(x,y)$ is a homogeneous on $y$--variables Finsler metric, see details in Refs. \cite{vncfinsl,vacpr,stavr0,girelli,lammer}. Generalizations of the principle of general covariance and axiomatics of  GR to $TV$ result in theories with arbitrary (nonhomogeneous) $F$ and more general metric structures $g_{\alpha \beta }(x,y)$ and frame transforms and
deformations.

A number of papers on Finsler gravity and
applications written by physicists are
 restricted only  to models with  ''nonlinear'' quadratic Finsler elements $%
ds^{2}=F^{2}(x,y)$ without important studies of physical implications of nonlinear and distinguished connection structures.  Non--experts in Finsler geometry consider that
locally anisotropic theories are completely defined by $F$ in a form which  is similar to  (pseudo)
Riemanian geometry which is  completely determined by a quadratic $\left(
~^{0}F\right) ^{2}=g_{ij}(x)y^{i}y^{j}$,  for $y^{i}\sim dx^{i}.$ Really, in GR a metric
tensor field $g_{ij}(x)$ defines a unique metric compatible Levi--Civita
connection $\nabla $ on $V,$ and $TV,$ and corresponding fundamental
Riemann/Ricci/Einstein tensors when the torsion field is constrained to be zero. Nevertheless, this is
not true for Finsler geometries and related gravity models because in such approaches the geometric constructions  are based on three fundamental geometric objects:\ a total
metric, $^{F}\mathbf{g,}$ a nonlinear connection, $^{F}\mathbf{N,}$ and a distinguished (adapted) linear connection, $^{F}\mathbf{D.}$ For certain
well--defined geometric/physical principles, all such values are
uniquely generated  by $F$ and this means  that a Finsler space\footnote{spacetime, if $^{F}\mathbf{g}
$ is related to a Minkowski metric in special relativity, or GR; for
simplicity, we omit the left label $F$ is this does not result in
ambiguities} is defined
by a triple of geometric data $\left( F:~\mathbf{g,N,D}\right)$.  Finsler theories are with more rich geometric structures than the
(pseudo) Riemannian ones determined by data $( \mathbf{g,\nabla }).$

In order to elaborate a self--consistent geometric model of classical and
quantum Finsler gravity theory we have to involve into constructions all
fundamental geometric/physical objects. Such values must be included into  certain  gravitational and matter field gravitational field equations (derived on $TM$ following certain generalized variational/geometric principles). It is also necessary to try to perform a quantization program and
then to analyze possible consequences/applications, for instance, in modern cosmology and astrophysics, or geometric mechanics, see details on such a series of works in Refs. \cite{vrevflg,vsgg,vacpr,vqgr1,vqgr2,vsing1}.

We must solve two important problems for  quantum/noncommutative Finsler generalizations of GR:
\begin{itemize}
\item What type of Finsler nonlinear and linear connections, $(\ ^{F}%
\mathbf{N},\ ^{F}\mathbf{D})$, are chosen following certain
 geometric and physical arguments? For instance, mathematicians %
\cite{bcs} prefer to work with the Chern and Berwald connections which are
metric noncompatible  (certain cosmological models \cite{lichang4,lichang6}
were elaborated following such an approach). Nevertheless, constructions with
metric noncompatible connections are less relevant to generalizations of
standard theories of particle physics because does not allow to
define in a usual form a particle classification, Dirac equations, conservation laws etc, see critical remarks in \cite{vcrit,vrevflg,vsgg}. In our works, we preferred to elaborate physical models when
 $\ ^{F}\mathbf{D}$ is chosen to  be the canonical distinguished
and/or normal/Cartan distinguished connections. Such constructions are metric compatible
and allow ''more standard'' theories of Finsler extension of GR (the
so--called Einstein--Finsler gravity, EFG, models).

\item Another problem is that if existing experimental data do not constrain
''too much'' the perspectives of Finsler gravity  for realistic QG and LV
theories? For instance, in Ref.  \cite{lammer}, such an analysis is
performed with the conclusion that coefficients $q_{\widehat{i}_{1}\widehat{i%
}_{2}...\widehat{i}_{2r}}$ in a Finsler metric (\ref{fbm}) and a related
dispersion relation  (\ref{disp}) seem to be very small and this sounds to
be very pessimistic for detecting a respective QG phenomenology and LV. Here we
note that a conclusion drown only using certain data for a Finsler metric $%
F(x,y)$  is not a final one because any parametrization  (\ref{fbm}) is
``geometric gauge`` dependent. Really, using frame/coordinate transforms and
nonholonomic deformations,
\begin{equation*}
\left( F:~\mathbf{g,N,D}\right) \rightarrow \left( ~^{0}F:~\mathbf{\check{g},%
\check{N},\check{D}}\right)
\end{equation*}%
when $\ ^{0}F$ is a typical quadratic form in GR, the LV effects are removed into data
$\left( \mathbf{\check{N},\check{D}}\right) $ modeling  nonlinear generic
off--diagonal quantum, and quasi--classical, interactions in QG. An explicit
example of such systems is that of noncommutative Finsler black holes  \cite%
{vncfinsl}. The black hole solutions and various
gravitational--gauge--fermion interactions, in GR and EFG can not be studied
experimentally only via Mikelson--Morley and possible related nonlinear
dispersion effects determined only by $F$. Off--diagonal metrics and anholonomic frames (nonlinear connection) and induced torsion effects (distinguished connection) effects are of crucial importance in Finsler theories.
\end{itemize}

A rigorous mathematical and physically motivated approach  should  consider formulations of Finsler
gravity theories for certain physically important classes of nonlinear and
distinguished connections. We should  try to find exact solutions and after that to
analyze possible physical implications not restricting our approach only to $%
F$ but to complete theories with nontrivial $\mathbf{N}$ and $\mathbf{D.}$
Surprisingly, such solutions can be constructed  in very general off--diagonal
forms \cite{vexsol,gensol2,vrevflg,vsgg} and we apply such methods in this paper. Nontrivial Finsler spaces on $TV$ are with generic
off--diagonal metrics $\mathbf{g}$ which in certain coordinate bases contain
contributions from $\mathbf{N.}$ Various limits from EFG to GR can be
modeled by a corresponding nonholonomic and nonlinear dynamics when the
coefficients of metrics depend anisotropically at least on 3--5--7 space,
time and velocity type coordinates on 8--d $TV$ Finsler spacetimes. Such scenarios
are more complex than the well--known compactification of extra dimensions
in Kaluza--Klein gravity. Even diagonal  metrics for  Finsler--Kaluza--Klein
gravity can be used for certain rough estimations, we need more sophisticate
classes of generic off--diagonal exact solutions with warping and trapping of interactions in
order to get a constant nonzero  value for the speed of light and generic
off--diagonal ''bulk'' configurations with nontrivial $\mathbf{N}$.

In this article,  we have constructed brane world solutions of gravitational
field equations for metric compatible EFG theories of QG and possible LV. We
found that using generic off--diagonal metrics, non--integrable constraints,
Finsler connections and stress--energy ansatz functions it is possible to
realize trapping gravitational configurations with physically resonable
properties for a range of parameters (for instance, the extra dimension, $%
m\geq 1;$ bulk cosmological constant $\Lambda ,$ in general, with locally
anisotropic polarizations; brane width $\epsilon ;$ a constant value $a$ of
gravitational interactions for the maximal speed of light/ gravitational
interactions etc). Such brane effects of QG with LV depend on the mechanism
of Finsler type gravitational and matter fields interactions on tangent
bundle $TV$ over a spacetime $V$ in general relativity (GR) and it is
expected that they may be detected in TeV physics, or via modifications in
modern cosmology and astrophysics (on locally anisotropic Finsler
cosmological scenarios and exact solutions with black ellipsoids, wormholes
etc see \cite{vsing1,vsgg}).

The solutions with trapping from $TV$ to a GR spacetime $V$ are of two
general forms:\ The first class consists from almost standard diagonal
generalizations/ modifications of results from Refs. \cite{midod,gm2} for
6-d (and higher) dimensions, when the Einstein equations for the
Levi--Civita connection were extended to 8--d (pseudo) Riemannian
spacetimes, with possible two time coordinates. If such diagonal brane
effects of QG and/or LV origin can be detected experimentally, we can
conclude that QG gravity is a (co) tangent bundle geometric theory for the
Levi--Civita connection determined by special types of nonlinear dispersions
(generating a trivial nonlinear connection, N--connection structure).

Nevertheless, very general assumptions on LV effects and nonlinear
dispersions induced from QG seem to result in a second class of Finsler like
nonlinear quadratic elements and canonically induced linear connections (for
instance, the so--called normal/Cartan d--connection) which are different
from the well known Levi--Civita connection. On $TV,$ it is naturally to
work with metric compatible d--connections with effective torsion completely
determined by a (Finsler) metric and N--connection structure as we discussed
in details in Refs. \cite{vrevflg,vcrit,vacpr,vsgg}. To construct brane
solutions with nontrivial N--connection structure and generic off--diagonal
metrics (the first attempts where considered in \cite{vsing1,vsgg}) is a
more difficult technical task which can be solved following the so--called
anholonomic deformation/frame method \cite{vexsol,gensol2}. We shall address
possible applications of nonholonomic geometry methods and Finsler brane
solutions in (non) commutative locally anisotropic cosmology and black holes
physics \cite{vncfinsl,vacpr,vqgr1}.

\vskip5pt

\textbf{Acknowledgements: } This paper contains some results presented at
Spanish Relativity Meeting, ERE2010, in Granada, Spain.

\appendix

\setcounter{equation}{0}\renewcommand{\theequation}
{A.\arabic{equation}} \setcounter{subsection}{0}
\renewcommand{\thesubsection}
{A.\arabic{subsection}}

\section{Einsten--Finsler Spaces of Dimension 2+2}

In this Appendix, we study a toy model of Einstein--Finsler gravity on $TM$
over a 2--dimensional manifold $M.$ We prove that such a theory can be
integrated in general form. Local coordinates are labeled $u^{\alpha
}=(x^{k},y^{a}),$ where indices run respectively the values: $%
i,j,k,...=1,2;\ a,b,c,...=3,4;\ $ and $y^{3}=v.$ Using frame transforms any
(pseudo) Finsler/ Riemannian 4--d metric can parametrized in the form
\begin{eqnarray}
\ \mathbf{g} &\mathbf{=}&g_{i}(x^{k})dx^{i}\otimes dx^{i}+\omega
^{2}(x^{j},y^{b})h_{a}(x^{k},v)\mathbf{e}^{a}{\otimes }\mathbf{e}^{a}
\label{killingdm} \\
&&\mbox{ for }\mathbf{e}^{3}=dy^{3}+w_{i}(x^{k},v)dx^{i},\mathbf{e}%
^{4}=dy^{4}+n_{i}(x^{k},v)dx^{i},  \notag
\end{eqnarray}%
which is a particular case of (\ref{dm}). We label in brief the partial
derivatives in the form $g_{1}^{\bullet }=\partial g_{1}/\partial
x^{1},g_{1}^{\prime }=\partial g_{1}/\partial x^{2}$ and $h_{3}^{\ast
}=\partial h_{3}/\partial v.$

For Finsler configurations, the condition of homogeneity results in at least
on Killing symmetry for metrics. We can always introduce such a N--adapted
frame/coordinate parametrization when $\omega ^{2}=1$ and the above metric
does not depend on variables $y^{4}.$ \ The coefficients of the normal/
Cartan d--connection $\widetilde{\mathbf{\Gamma }}_{\ \alpha \ \beta }^{\
\gamma }$ (\ref{cdc}) can be computed for a metric (\ref{killingdm}) with $%
\omega ^{2}=1,$ when $g_{\alpha \beta }=diag[g_{i}(x^{k}),h_{a}(x^{i},v)]$
and $N_{k}^{3}=w_{k}(x^{i},v),N_{k}^{4}=n_{k}(x^{i},v).$\footnote{%
Following methods elaborated in Refs. \cite{gensol2,vexsol}, we can
construct exact solutions with $\omega ^{2}\neq 1.$ For Finsler brane
configurations, for simplicity, we do not consider such "very" general
classes of solutions.} Using  the Cartan structure equations (\ref{cartseq}%
), it is possible to determine the $h$-- and $v$--components of the
Riemannian, torsion, Ricci d--tensors etc.


For the 2+2 dimensional EFG theory, there is a very important property of
decoupling/separation of field equations with respect to a class of
N--adapted frames which allows us to integrate the theory in very general
forms (see Theorem \ref{th4}) depending on the types of prescribed
nonholonomic constraints and given sources parametrized by frame transform
as {\small
\begin{equation}
\widetilde{\mathbf{\Upsilon }}_{\beta }^{\ \delta }= diag[ \widetilde{%
\mathbf{\Upsilon }}_{1}^{\ 1}= \widetilde{\mathbf{\Upsilon }}_{2}^{\ 2}=\
^{v}\Lambda (x^{i},v),\widetilde{\mathbf{\Upsilon }}_{3}^{\ 3}= \widetilde{%
\mathbf{\Upsilon }}_{4}^{\ 4}=\ ^{h}\Lambda (x^{i})].  \label{sourceans}
\end{equation}%
} As particular cases, such sources generalize contributions from nontrivial
cosmological constants (for instance, if $\ ^{h}\Lambda =\ ^{v}\Lambda
=\Lambda =const$), their nonholonomic matrix polarizations, approximations
for certain dust/radiation locally anisotropic states of matter etc.

\begin{theorem}
\label{th3}The Finsler gravitational field equations (\ref{ensteqcdc}) for a
metric (\ref{killingdm}) with $\omega ^{2}=1$ and source (\ref{sourceans})
are equivalent to this system of partial differential equations: {\small
\begin{eqnarray}
\widetilde{R}_{1}^{1} &=&\widetilde{R}_{2}^{2} = \frac{1}{2g_{1}g_{2}}[\frac{%
g_{1}^{\bullet }g_{2}^{\bullet }}{2g_{1}}+\frac{(g_{2}^{\bullet })^{2}}{%
2g_{2}}-g_{2}^{\bullet \bullet }+\frac{g_{1}^{^{\prime }}g_{2}^{^{\prime }}}{%
2g_{2}}+\frac{(g_{1}^{^{\prime }})^{2}}{2g_{1}}-g_{1}^{^{\prime \prime
}}]=-\ ^{h}\Lambda,  \label{4ep1a} \\
\widetilde{R}_{3}^{3} &=&\widetilde{R}_{4}^{4}=\frac{1}{2h_{3}h_{4}}\left[
-h_{4}^{\ast \ast }+\frac{\left( h_{4}^{\ast }\right) ^{2}}{2h_{4}}+\frac{%
h_{3}^{\ast }h_{4}^{\ast }}{2h_{3}}\right] =\ -\ ^{v}\Lambda,  \label{4ep2a}
\\
\ \widetilde{R}_{3j} &=&\frac{h_{3}^{\ast }}{2h_{3}}w_{j}^{\ast }+A^{\ast
}w_{j}+B_{j}=0,  \label{4ep3a} \\
\widetilde{R}_{4i} &=&-\frac{h_{4}^{\ast }}{2h_{3}}n_{i}^{\ast }+\frac{%
h_{4}^{\ast }}{2}K_{i}=0,  \label{4ep4a} \\
A&=&\left( \frac{h_{3}^{\ast }}{2h_{3}}+\frac{h_{4}^{\ast }}{2h_{4}}\right)
,\ B_{k}=\frac{h_{4}^{\ast }}{2h_{4}}\left( \frac{\partial _{k}g_{1}}{2g_{1}}%
-\frac{\partial _{k}g_{2}}{2g_{2}}\right) -\partial _{k}A,  \label{aux} \\
K_{1} &=&-\frac{1}{2}\left( \frac{g_{1}^{\prime }}{g_{2}h_{3}}+\frac{%
g_{2}^{\bullet }}{g_{2}h_{4}}\right) ,\ K_{2}=\frac{1}{2}\left( \frac{%
g_{2}^{\bullet }}{g_{1}h_{3}}-\frac{g_{2}^{\prime }}{g_{2}h_{4}}\right) .
\notag
\end{eqnarray}
}
\end{theorem}

\begin{proof}
We apply the constructions for the canonical d--connection from \cite%
{vexsol,gensol2,vsgg} to the case of normal/ Cartan d--connection on 4--d $%
TM.$ Following definition of coefficients $\widetilde{\mathbf{\Gamma }}_{\
\alpha \ \beta }^{\ \gamma }$ (\ref{cdc}), the $h$-- and $v$--components are
similar to those for those for the canonical d--connection. So, the proofs
for equations (\ref{4ep1a}) and (\ref{4ep2a}) are completely similar to
those for presented in \ the mentioned works. For $\widetilde{\mathbf{\Gamma
}}_{\ \alpha \ \beta }^{\ \gamma },$ there are differences for (\ref{4ep3a})
and (\ref{4ep4a}) which are analyzed in this Appendix.

We perform a N--adapted differential calculus if instead of partial
derivatives $\partial _{\ \alpha }=\partial /\partial u^{\ \alpha }$ there
are considered operators (\ref{nader}) parametrized in the form $\mathbf{e}\
_{i}=\partial _{\ i}-N_{\ i}^{\ a}\partial _{\ a}=\partial _{\ i}-w_{i\ }^{\
}\partial _{v}-n_{i}^{\ }\partial _{4\ }.$ For $%
N_{k}^{3}=w_{k}(x^{i},v),N_{k}^{4}=n_{k}(x^{i},v),$ the nontrivial
coefficients of N--connection curvature are
\begin{equation}
\Omega _{\ 12}^{\ 3}=w_{2}^{\ \bullet }-w_{\ 1}^{\prime \ }-w_{1\ }^{\
}w_{2}^{\ \ast }+w_{2\ }^{\ }w_{1}^{\ \ast },\Omega _{\ 12}^{\
4}=n_{2}^{\bullet \ }-n_{1}^{\prime \ }-w_{1\ }^{\ }n_{2}^{\ \ast }+w_{2\
}^{\ }n_{1}^{\ \ast }.  \label{nccan}
\end{equation}
There are nontrivial coefficients of $\ \widetilde{\mathbf{\Gamma }}_{\ \
\alpha \ \beta }^{\ \gamma }$, {\small
\begin{eqnarray}
\widetilde{L}_{11}^{1} &=&\frac{g_{1}^{\bullet }}{2g_{1}},\ \widetilde{L}%
_{12}^{1}=\frac{g_{1}^{\prime }}{2g_{1}},\widetilde{L}_{22}^{1}=-\frac{%
g_{2}^{\bullet }}{2g_{1}},\widetilde{L}_{11}^{2}=-\frac{g_{1}^{\prime }}{%
2g_{2}},\ \widetilde{L}_{12}^{2}=\frac{g_{2}^{\bullet }}{2g_{2}},\   \notag
\\
\ \ \widetilde{L}_{22}^{2} &=&\frac{g_{2}^{\prime }}{2g_{2}},\widetilde{C}%
_{33}^{3}=\frac{h_{3}^{\ast }}{2h_{3}},\ \widetilde{C}_{44}^{3}=-\frac{%
h_{4}^{\ast }}{2h_{3}},\ \widetilde{C}_{34}^{4}=\frac{h_{4}^{\ast }}{2h_{4}}.
\label{auxcartc}
\end{eqnarray}%
} The nontrivial coefficients of torsion (\ref{nztors}) are
\begin{eqnarray*}
\widetilde{T}_{\ 12}^{\ 3} &=&\Omega _{\ 21}^{\ 3},\ \widetilde{T}_{\ 12}^{\
4}=\Omega _{\ 21}^{\ 4}, \widetilde{P}_{i3}^{3} = w_{i}^{\ast }-\frac{%
\partial _{i}g_{1}}{2g_{1}}; \widetilde{P}_{\ 14}^{3}=-\frac{g_{1}^{\prime }%
}{2g_{1}},\widetilde{P}_{\ 24}^{3}=\frac{g_{2}^{\bullet }}{2g_{1}}, \\
\widetilde{P}_{\ 13}^{4} &=& n_{1}^{\ast }+\frac{g_{1}^{\prime }}{2g_{2}},
\widetilde{P}_{\ 23}^{4}=n_{2}^{\ast }-\frac{g_{2}^{\bullet }}{2g_{1}};%
\widetilde{P}_{i4}^{4}=-\frac{\partial _{i}g_{2}}{2g_{2}}.
\end{eqnarray*}

The h--v components of the Ricci tensor are derived from
\begin{eqnarray*}
\widetilde{R}_{\ bka}^{c} &=&\frac{\partial \widetilde{L}_{.bk}^{c}}{%
\partial y^{a}}-\widetilde{C}_{.ba|k}^{c}+\widetilde{C}_{.bd}^{c}\widetilde{P%
}_{.ka}^{d} \\
&=&\frac{\partial \widetilde{L}_{.bk}^{c}}{\partial y^{a}}-(\frac{\partial
\widetilde{C}_{.ba}^{c}}{\partial x^{k}}+\widetilde{L}_{.dk}^{c\,}\widetilde{%
C}_{.ba}^{d}-\widetilde{L}_{.bk}^{d}\widetilde{C}_{.da}^{c}-\widetilde{L}%
_{.ak}^{d}\widetilde{C}_{.bd}^{c}) +\widetilde{C}_{.bd}^{c}\widetilde{P}%
_{ka}^{d}.
\end{eqnarray*}%
Contracting indices, we get $\widetilde{R}_{bk}=\widetilde{R}_{\ bka}^{a}=%
\frac{\partial \widetilde{L}_{.bk}^{a}}{\partial y^{a}}-\widetilde{C}_{b|k}+%
\widetilde{C}_{.bd}^{a}\widetilde{P}_{ka}^{d}$, where $\widetilde{C}_{b}=%
\widetilde{C}_{.ba}^{c}$ and $\partial \widetilde{L}_{.bk}^{a}/\partial
y^{a}=0$ for (\ref{killingdm}) with $\omega ^{2}=1.$ We have
\begin{eqnarray*}
\widetilde{C}_{b|k} &=&\mathbf{e}_{k}\widetilde{C}_{b}-\widehat{L}_{\
bk}^{d\,}\widetilde{C}_{d}=\partial _{k}\widetilde{C}_{b}-N_{k}^{e}\partial
_{e}\widetilde{C}_{b}-\widetilde{L}_{\ bk}^{d}\widetilde{C}_{d} \\
&=&\partial _{k}\widetilde{C}_{b}-w_{k}\widetilde{C}_{b}^{\ast }-\widetilde{L%
}_{\ bk}^{d\,}\widetilde{C}_{d},
\end{eqnarray*}%
for $\widetilde{C}_{3}=\widetilde{C}_{33}^{3}+\widetilde{C}_{34}^{4}=\frac{%
h_{3}^{\ast }}{2h_{3}}+\frac{h_{4}^{\ast }}{2h_{4}},\ \widetilde{C}_{4}=%
\widetilde{C}_{43}^{3}+\widetilde{C}_{44}^{4}=0$, see (\ref{auxcartc}).

We express $\widetilde{R}_{bk}=\ ^{1}\widetilde{R}_{bk}+\ ^{2}\widetilde{R}%
_{bk}+\ ^{3}\widetilde{R}_{bk},$ where%
\begin{eqnarray*}
\ ^{1}\widetilde{R}_{bk} &=&\left( \widetilde{L}_{bk}^{4}\right) ^{\ast
}=0,\ ^{2}\widetilde{R}_{bk}=-\ \widetilde{C}_{b|k}=-\partial _{k}\
\widetilde{C}_{b}+w_{k}\ \widetilde{C}_{b}^{\ast }+\widetilde{L}_{\
bk}^{d\,}\ \widetilde{C}_{d}, \\
\ \ ^{3}\widetilde{R}_{bk} &=&\ \widetilde{C}_{bd}^{a}\widetilde{P}%
_{.ka}^{d}=\ \widetilde{C}_{b3}^{3}\widetilde{P}_{k3}^{3}+\ \widetilde{C}%
_{b4}^{3}\widetilde{P}_{k3}^{4}+\ \widetilde{C}_{b3}^{4}\widetilde{P}%
_{k4}^{3}+\ \widetilde{C}_{b4}^{4}\widetilde{P}_{k4}^{4}.
\end{eqnarray*}%
Then, it is possible to compute $\widetilde{R}_{3k}=\ ^{2}\widetilde{R}%
_{3k}+\ ^{3}\widetilde{R}_{3k}$ when, for instance, $\widehat{L}_{\
3k}^{3\,}\rightarrow \widehat{L}_{\ 1k}^{1\,}$ and $\widehat{L}_{\
4k}^{4\,}\rightarrow \widehat{L}_{\ 2k}^{2\,},\ $ with {\small
\begin{eqnarray*}
\ ^{2}\widetilde{R}_{3k} &=&-\partial _{k}\widetilde{C}_{3}+w_{k}\widetilde{C%
}_{3}^{\ast }+\widehat{L}_{\ 3k}^{3\,}\widetilde{C}_{3} \\
&=& -\partial _{k}\left( \frac{h_{3}^{\ast }}{2h_{3}}+\frac{h_{4}^{\ast }}{%
2h_{4}}\right) +\left( \frac{h_{3}^{\ast }}{2h_{3}}+\frac{h_{4}^{\ast }}{%
2h_{4}}\right) ^{\ast }w_{k}+\left( \frac{h_{3}^{\ast }}{2h_{3}}+\frac{%
h_{4}^{\ast }}{2h_{4}}\right) \frac{\partial _{k}g_{1}}{2g_{1}}, \\
\ ^{3}\widetilde{R}_{3k} &=&\ \ \widetilde{C}_{33}^{3}\widetilde{P}%
_{k3}^{3}+\ \widetilde{C}_{34}^{3}\widetilde{P}_{k3}^{4}+\ \widetilde{C}%
_{33}^{4}\widetilde{P}_{k4}^{3}+\ \widetilde{C}_{34}^{4}\widetilde{P}%
_{k4}^{4} \\
&=& \frac{h_{3}^{\ast }}{2h_{3}}\left( w_{i}^{\ast }-\frac{\partial _{i}g_{1}%
}{2g_{1}}\right) -\frac{h_{4}^{\ast }}{2h_{4}}\frac{\partial _{i}g_{2}}{%
2g_{2}}.
\end{eqnarray*}%
} (\ref{4ep3a}) can be obtained summarizing above formulas, {\small
\begin{equation*}
\widetilde{R}_{3k}=\frac{h_{3}^{\ast }}{2h_{3}}w_{i}^{\ast }+\left( \frac{%
h_{3}^{\ast }}{2h_{3}}+\frac{h_{4}^{\ast }}{2h_{4}}\right) ^{\ast
}w_{k}-\partial _{k}\left( \frac{h_{3}^{\ast }}{2h_{3}}+\frac{h_{4}^{\ast }}{%
2h_{4}}\right) +\frac{h_{4}^{\ast }}{2h_{4}}\left( \frac{\partial _{k}g_{1}}{%
2g_{1}}-\frac{\partial _{k}g_{2}}{2g_{2}}\right).
\end{equation*}
}

Similarly, we compute $\ \widetilde{R}_{4k}=\ ^{2}\widetilde{R}_{4k}+\ ^{3}%
\widetilde{R}_{4k},$ where {\small
\begin{eqnarray*}
\ \ ^{2}\widetilde{R}_{4k} &=&-\partial _{k}\ \widetilde{C}_{4}+w_{k}%
\widetilde{C}_{4}^{\ast }+\widetilde{L}_{\ 4k}^{3\,}\ \widetilde{C}_{4}=0; \\
\ \ ^{3}\widetilde{R}_{4k} &=&\ \widetilde{C}_{43}^{3}\widetilde{P}%
_{k3}^{3}+\ \widetilde{C}_{44}^{3}\widetilde{P}_{k3}^{4}+\ \widetilde{C}%
_{43}^{4}\widetilde{P}_{k4}^{3}+\ \widetilde{C}_{44}^{4}\widetilde{P}%
_{k4}^{4} = \widetilde{C}_{44}^{3}\widetilde{P}_{k3}^{4}+ \widetilde{C}%
_{43}^{4}\widetilde{P}_{k4}^{3}.
\end{eqnarray*}%
} Putting together, we obtain (\ref{4ep4a}) [which ends the proof of Theorem %
\ref{th3}],
\begin{equation*}
\ \widetilde{R}_{41}=-\frac{h_{4}^{\ast }}{2h_{3}}(n_{1}^{\ast }+\frac{%
g_{1}^{\prime }}{2g_{2}}) -\frac{g_{2}^{\bullet }}{2g_{2}}\frac{h_{4}^{\ast }%
}{2h_{4}},\ \widetilde{R}_{42}=-\frac{h_{4}^{\ast }}{2h_{3}}(n_{2}^{\ast }-%
\frac{g_{2}^{\bullet }}{2g_{1}}) -\frac{g_{2}^{\prime }}{2g_{2}}\frac{%
h_{4}^{\ast }}{2h_{4}}.
\end{equation*}
\end{proof}

\section{Integration of field equations}

\begin{theorem}
\label{th4}The general solutions of equations (\ref{4ep1a})--(\ref{4ep4a})
defining Ein\-stein--Finsler spaces are parametrized by metrics of type (\ref%
{killingdm}) with coefficients computed in the form {\small
\begin{eqnarray}
g_{i} &=&\epsilon _{i}e^{\psi (x^{k})},\mbox{\ for }\epsilon _{1}\psi
^{\bullet \bullet }+\epsilon _{2}\psi ^{\prime \prime }=\ ^{h}\Lambda
(x^{k});  \label{sol1} \\
h_{3} &=&\epsilon _{3}\ ^{0}h(x^{i})\ [f^{\ast }(x^{i},v)]^{2}|\varsigma
(x^{i},v)|,  \label{sol2} \\
\varsigma &=&\ ^{0}\varsigma (x^{i})-\frac{\epsilon _{3}}{8}\
^{0}h(x^{i})\int(dv)\ ^{v}\Lambda (x^{k},v) f^{\ast }(x^{i},v)\
[f(x^{i},v)-\ ^{0}f(x^{i})],  \notag \\
h_{4} &=&\epsilon _{4}[f(x^{i},v)-\ ^{0}f(x^{i})]^{2};  \label{sol3} \\
w_{j} &=&\ _{0}w_{j}(x^{i})\exp \left\{ -\int_{0}^{v}[ {2h_{3}A^{\ast }}/{%
h_{3}^{\ast }}] _{v\rightarrow v_{1}}dv_{1}\right\} \times  \label{sol4} \\
&&\int_{0}^{v}dv_{1}[{h_{3}B_{j}}/{h_{3}^{\ast }}] _{v\rightarrow v_{1}}\exp
\left\{ -\int_{_{0}}^{v_{1}}[{2h_{3}A^{\ast }}/{h_{3}^{\ast }}]
_{v\rightarrow v_{1}}dv_{1}\right\} ,  \notag \\
n_{i} &=&\ _{0}n_{i}(x^{k})+\int dv\ h_{3}K_{i}.  \label{sol4a}
\end{eqnarray}%
} Such solutions with $h_{3}^{\ast }, h_{4}^{\ast }\neq 0$ are determined by
generating, $f(x^{i},v),f^{\ast }\neq 0,$ and integration, $\ ^{0}f(x^{i}),\
^{0}h(x^{i}),$ $\ _{0}w_{j}(x^{i}),\ _{0}n_{i}(x^{k}),$ functions.
\end{theorem}

\begin{proof}
We sketch a proof following two steps:

\begin{enumerate}
\item \textbf{Solutions with Killing symmetry for h-- and v--components of
metric:} The equation (\ref{4ep1a}) is for a two dimensional (semi)
Riemannian metric. Any such metric can be diagonalized and expressed as a
conformally flat metric. Choosing $\epsilon _{i}e^{\psi (x^{k})},$ we get
the Poisson equation in (\ref{sol1}). The equation (\ref{4ep2a}) is similar
to that for the canonical d--connection configurations which was solved in
general form \cite{vexsol,gensol2,vsgg}. Such equations relate two un--known
functions. For instance, if we prescribe any $h_{3}(x^{i},v),$ we can
construct (at least via some series decompositions) $h_{4}(x^{i},v),$ and
inversely. By straightforward computations, we can verify that any $h_{3}$
and $h_{4}$ with nonzero $h_{3}^{\ast }$ and $h_{4}^{\ast }$ given by (\ref%
{sol2}) define exact solutions for (\ref{4ep2a}). Solutions with $%
h_{3}^{\ast }=0$ and/or $h_{4}^{\ast }=0$ should be re--considered as some
particular degenerated cases.

\item \textbf{Solutions for the N--connection coefficients: } The main
differences between our former results for the canonical d--connection and
the normal/ Cartan d--connection (in this work) consist in equations (\ref%
{4ep3a}) and (\ref{4ep4a}) and coefficients (\ref{aux}). We provide the
proofs of formulas (\ref{sol3}) and (\ref{sol4a})) in Appendix \ref{assc}.
Taking together the solutions (\ref{sol1})--(\ref{sol4a}) for ansatz (\ref%
{killingdm}) with $\omega ^{2}=1,$ we constrict the general class of exact
solutions with Killing symmetry on $\partial /\partial y^{4}$ defining
Einstein--Finsler spaces\footnote{%
such a symmetry exists if the coefficients of metrics do not depend on
coordinate $y^{4}$}. Considering different types of frame transforms, with
coordinates parametrized for tangent bundles, such metrics can be
transformed into standard ones $\ ^{F}\mathbf{g}$ (\ref{slm}) Finsler spaces.
\end{enumerate}
\end{proof}

\textbf{Some computations for Theorem \ref{th4}:}

\label{assc}The solutions of (\ref{4ep3a}) and (\ref{4ep4a}) can be always
considered for $|g_{1}|=|g_{2}|,$ when $B_{k}=\partial _{k}A.$ We construct
them for three more special cases.

Case 1: $h_{3}^{\ast }=0,h_{4}^{\ast }\neq 0$ and $A=h_{4}^{\ast }/2h_{4}.$
We must solve the equation $h_{4}^{\ast \ast }-\frac{\left( h_{4}^{\ast
}\right) ^{2}}{2h_{4}}=2h_{3}h_{4}\ ^{v}\Lambda (x^{i},v),$ for any given $%
h_{3}=h_{3}(x^{i})$ and $\ ^{v}\Lambda (x^{i},v).$ We have $w_{j}^{\ast }=0$
and we obtain, from (\ref{4ep3a}), $w_{j}=-B_{j}/A^{\ast }=-\partial
_{j}A/A^{\ast }$ and, from (\ref{4ep4a}), $n_{i}^{\ast }=K_{i}h_{3}.$

Case 2: $h_{4}^{\ast }=0,$ any $h_{3}$ and $n_{i}$ for $\ ^{v}\Lambda =0.$
Let us consider in (\ref{4ep3a}) that $h_{3}\neq 0.$ We have to solve $\frac{%
h_{3}^{\ast }}{2h_{3}}w_{j}^{\ast }+A^{\ast }w_{j}+B_{j}=0.$ Representing $%
w_{j}=\ ^{1}w_{j}\cdot \ ^{2}w_{j}$ and introducing $\ w_{j}^{\ast }=\
^{1}w_{j}^{\ast }\cdot \ ^{2}w_{j}+\ ^{1}w_{j}\cdot \ ^{2}w_{j}^{\ast }$
into above equation, we obtain
\begin{equation*}
\ ^{1}w_{j}^{\ast }\cdot \ ^{2}w_{j}+\ ^{1}w_{j}\cdot \ ^{2}w_{j}^{\ast }+%
\frac{2h_{3}A^{\ast }}{h_{3}^{\ast }}\ ^{1}w_{j}\cdot \ ^{2}w_{j}+\frac{%
2h_{3}B_{j}}{h_{3}^{\ast }}=0.
\end{equation*}%
We can chose $\ ^{1}w_{j}=-\ _{0}^{1}w_{j}(x^{i})\exp \left[ -\int \frac{%
2h_{3}A^{\ast }}{h_{3}^{\ast }}dv\right] ,$ for some integration functions $%
\ _{0}^{1}w_{j}(x^{i}),$ and transform the equation into$\ ^{2}w_{j}^{\ast
}=2\ _{0}^{1}w_{j}(x^{i})$\newline
$\int_{v_{2}}^{v}dv_{1}\frac{h_{3}B_{j}}{h_{3}^{\ast }}\exp \left[
-\int_{v_{0}}^{v_{1}}\frac{2h_{3}A^{\ast }}{h_{3}^{\ast }}dv_{1}\right] $,
which can be integrated in general form. Finally, we get
\begin{equation*}
w_{j}=\ _{0}w_{j}(x^{i})\exp \left[ -\int \frac{2h_{3}A^{\ast }}{h_{3}^{\ast
}}dv\right] \int_{v_{2}}^{v}dv_{1}\frac{h_{3}B_{j}}{h_{3}^{\ast }}\exp \left[
-\int_{v_{0}}^{v_{1}}\frac{2h_{3}A^{\ast }}{h_{3}^{\ast }}dv_{1}\right] ,
\end{equation*}%
for some $v_{0},v_{2}=const.$ The solution of (\ref{4ep4a}) is constructed
by a direct integration on $v$ of values $K_{i}$ from (\ref{aux}).

Case 3: $h_{3}^{\ast }\neq 0,h_{4}^{\ast }\neq 0,$ which is stated in
Theorem \ref{th4}. As a general solution of (\ref{4ep3a}), we can consider (%
\ref{sol3}) with the coefficients $A,B_{j}$ and $K_{i}$ and computed for
arbitrary $h_{3}$ and $h_{4}$ depending on $v.$ Integrating (\ref{4ep4a}),
we get the formula (\ref{sol4a}).

\end{document}